\documentclass[aps,twocolumn,floatfix,superscriptaddress,longbibliography,normalem]{revtex4-2}

\usepackage{adjustbox}
\usepackage{amssymb,amsmath,amstext,amsthm}
\DeclareFontFamily{U}{mathx}{}
\DeclareFontShape{U}{mathx}{m}{n}{<-> mathx10}{}
\DeclareSymbolFont{mathx}{U}{mathx}{m}{n}
\DeclareMathAccent{\widecheck}{0}{mathx}{"71}
\usepackage{dsfont}
\usepackage{graphicx}
\usepackage{bm}
\usepackage{algorithm,algpseudocode}
\usepackage[T1]{fontenc}
\usepackage{physics}
\usepackage{xcolor}
\usepackage[colorlinks=true,citecolor=blue,linkcolor=magenta]{hyperref}
\usepackage{tikz-network}
\usepackage{orcidlink}
\usepackage{float}
\definecolor{claudeorange}{HTML}{9E4228}

\newtheorem{definition}{Definition}

\newtheorem{informalalg}[algorithm]{Algorithm}

\theoremstyle{plain}
\newtheorem{proposition}{Proposition}

\newtheorem{theorem}[proposition]{Theorem}
\newtheorem{corollary}[proposition]{Corollary}
\theoremstyle{definition}
\theoremstyle{remark}

\newcommand{\quantummotionaus}{Quantum Motion, 4 Cornwallis St, Eveleigh NSW 2015, Australia}
\newcommand{\quantummotionuk}{Quantum Motion, 9 Sterling Way, London N7 9HJ, United Kingdom}
\newcommand{\unimelb}{School of Physics, The University of Melbourne, Parkville, VIC 3010, Australia}
\newcommand{\oxfordmath}{Mathematical Institute, University of Oxford, Woodstock Road, Oxford OX2 6GG, United Kingdom}

\begin{document}
\title{Breaking the chain: geometry-native state preparation with ASPIRE}

\author{Fredrik~Hasselgren\,\orcidlink{0009-0007-9273-4480}}
\affiliation{\quantummotionuk}
\affiliation{\oxfordmath}
\author{Matthew~L.~Sims-Goh\,\orcidlink{0000-0002-7478-4026}}
\email{matt.sims-goh@quantummotion.tech}
\affiliation{\quantummotionaus}
\affiliation{\unimelb}

\begin{abstract}
Loading data into quantum computers is a key bottleneck that threatens end-to-end quantum advantage. Quantum-circuit compilers for structured states like matrix product states (MPS) are a critical enabler to manage this in quantum chemistry, finance, and many-body quantum simulation. Approximate preparation of MPS typically relies on layered sequential staircases of nearest-neighbour gates. We present Adaptive State Preparation by Iterative Removal of Entanglement (ASPIRE), an algorithm that approximately compiles state-preparation circuits by selecting gates based on the state's underlying entanglement geometry and the available connectivity of the target quantum processor. Our protocol scores candidate qubit pairs on the removability of their entanglement and identifies layers of parallelised long-ranged gates that remove the most entanglement. We demonstrate our protocol empirically across a range of settings, illustrating its success for preparing Hamiltonian ground states and multivariate amplitude-encoded functions, and accounting for qubit connectivity, hardware noise, and both NISQ and early fault-tolerant settings. Throughout, ASPIRE typically achieves higher fidelities and shallower circuits than existing methods for states with long-ranged entanglement, even when classical optimisation is accounted for. Our investigation positions ASPIRE as a competitive, resource-frugal state-preparation protocol, with particular promise for applications tolerant to imperfect fidelity such as preparation of guide states for quantum phase estimation.
\end{abstract}

\maketitle

\section{Introduction\label{sec:introduction}}

Quantum computing promises substantial advantage on a range of high-impact problems, but end-to-end quantum advantage is easily thwarted by state preparation bottlenecks.
For instance, studying Hamiltonian ground states with quantum phase estimation (QPE) requires high-overlap guide states~\cite{tubman2018postponing, lee2023evaluating, ollitrault2024enhancing, marti2025spin, fomichev2024initial, berry2025rapid, erle2026shadow}, solving exponentially-large linear systems demands efficient loading of the input vector~\cite{aaronson2015read}, and pricing financial derivatives with quantum amplitude estimation requires one to load the underlying probability distribution~\cite{woerner2019quantum, stamatopoulos2020option}.
Because specifying a generic quantum state requires exponentially many independent amplitudes, loading arbitrary quantum states is exponentially costly~\cite{shende2005synthesis}, which threatens to eliminate end-to-end asymptotic quantum advantage regardless of the speed-up provided by the quantum algorithm itself.

Efficient quantum state preparation must therefore exploit some underlying structure constraining the amplitudes by an asymptotically smaller number of free parameters, often connecting to some notion of classical simulability~\cite{grover2002creating,kitaev2008wavefunction,kivlichan2018quantum,zoufal2019quantum,ran2020encoding,holmes2020efficient,rattew2021efficient,rattew2022preparing,markov2022generalized,mcardle2022quantum,moosa2023linear,marin2023quantum,rudolph2023decomposition,mori2024efficient,gonzalez2024efficient,fomichev2024initial,iaconis2024quantum,rosenkranz2025quantum,goh2025lie,berry2025rapid,sugawara2025embedding,huggins2025efficient,rupprecht2026sparse,bohun2026entanglement,barligea2026enabling}.
One particularly promising such avenue is the use of matrix product states (MPS) for state preparation.
These states can efficiently represent quantum many-body systems with area-law entanglement \cite{eisert2010colloquium}, and have further been applied to a myriad of other applications like molecular ground states~\cite{chan2011density,baiardi2020density,menczer2024tensor,menczer2024parallel,xiang2024distributed,legeza2025orbital}, 2D condensed-matter systems \cite{white2007neel,yan2011spin,stoudenmire2012studying,zheng2017stripe}, multivariate functions \cite{oseledets2010approximation,khoromskij2011quantics,garcia2021quantum,ye2022quantum,gourianov2022quantum,shinaoka2023multiscale,ritter2024quantics,nunez2025learning,bohun2026entanglement,sims2026multivariate,lee2026quantum,connor2026tensor,niedermeier2026solving}, loading genomic data into quantum processors \cite{creevey2025scalable}, and the preparation of guide states for ground state computations with QPE~\cite{fomichev2024initial, berry2025rapid}.

Crucially, bringing an MPS into canonical form naturally implies a physical state preparation scheme, since each local tensor is an isometry that directly maps virtual `bond' degrees of freedom onto physical sites \cite{schoen2005sequential}.
Each isometry is physically realisable by unitary operations (i.e. equivalent to a unitary where only certain columns need act faithfully), enabling one to compile a state-preparation circuit for an MPS directly from its canonical form.
A diverse array of techniques has emerged for this MPS-to-circuit compilation, beginning with the sequential `staircase' of Sch\"{o}n \emph{et al.}~\cite{schoen2005sequential}, and later growing to include direct synthesis of isometries~\cite{iten2016quantum,iten2019introduction,malvetti2021quantum}, QROM-based approaches for fault-tolerant systems~\cite{fomichev2024initial,berry2025rapid,rupprecht2026sparse}, logarithmic-depth renormalisation-group-based constructions \cite{malz2024preparation,scheer2025renormalization,murota2026exact,mingare2026practical}, and methods exploiting mid-circuit measurement and feed-forward adaptive circuits~\cite{smith2024constant,zhang2024characterizing,malz2024preparation,stephen2025preparing}, each with comparative advantages and useful regimes.

However, these algorithms typically aim for exact preparation of the MPS. For most practical state-preparation applications, the underlying MPS of interest is already approximate, such as an MPS approximation to a Hamiltonian ground state obtained from the density matrix renormalisation group (DMRG) \cite{white1992density,white1993density,schollwock2011density} or a multivariate function approximated with tensor cross-interpolation (TCI)~\cite{ritter2024quantics,nunez2025learning,tindall2024compressing,sims2026multivariate}, both of which are limited by controllable truncation errors. Furthermore, many algorithms tolerate an imperfect initial state gracefully --- for example, in QPE a small total infidelity $\varepsilon$ between the guide state and the true target merely adds a small penalty to the repetition count scaling as ${\sim}(1+ 2\varepsilon)$, rather than compromising the success of the end-to-end procedure. In these cases, it is wasteful to spend excessive quantum resources on \emph{exactly} compiling a state to a level of accuracy that exceeds its inherent approximations or the downstream quantum algorithm's error tolerance. This motivates approximate methods with controllable trade-offs between fidelity and circuit costs.

Such methods emerged with the `matrix product disentangler' (MPD) of Ran~\cite{ran2020encoding}, which iteratively `disentangles' a target MPS with repeated sequential layers of 2-qubit gates, constructing an approximate state-preparation circuit. Subsequent work extended this by introducing fidelity-based gate optimisation~\cite{zhou2021automatically,rudolph2023decomposition,melnikov2023quantum,ben2024approximate}, applying the framework to amplitude-encoded functions and classical datasets~\cite{melnikov2023quantum,green2025quantum,bohun2026entanglement,sims2026multivariate,szoldra2026scalable}, recasting the disentangling objective as a local entanglement minimisation~\cite{green2025quantum,mansuroglu2026preparation}, and extending the methodology to more sophisticated tensor-network topologies such as tree tensor networks~\cite{sugawara2025embedding,sims2026multivariate} and more general tensor networks tractable with belief propagation~\cite{szoldra2026belief}.

Although MPS-based methods rely upon an inherently 1D representation, they have been surprisingly effective for systems beyond 1D --- DMRG has been successful in estimating ground states of molecules~\cite{chan2011density,baiardi2020density,menczer2024tensor,menczer2024parallel,xiang2024distributed,legeza2025orbital} and 2D condensed-matter systems~\cite{white2007neel,yan2011spin,stoudenmire2012studying,zheng2017stripe}, simulating quantum circuits with 2D connectivity~\cite{zhou2020limits}, and representing multivariate functions where both inter- and intra-variable bit correlations preclude an optimal 1D ordering in general~\cite{garcia2021quantum,ye2022quantum,gourianov2022quantum,shinaoka2023multiscale,ritter2024quantics,nunez2025learning,bohun2026entanglement,sims2026multivariate,lee2026quantum,connor2026tensor,niedermeier2026solving}. With modern accelerated computing techniques pushing the limits of available bond dimension~\cite{ganahl2023density,menczer2024tensor,menczer2024parallel,xiang2024distributed,legeza2025orbital}, these capabilities to represent non-1D systems are only growing. Although targeting reduction in \emph{local} entanglement at individual bonds of the 1D MPS has proven a successful strategy for enhancing MPD methods~\cite{green2025quantum,mansuroglu2026preparation}, prior work in MPD methods has largely focused on 1D circuit structures for state preparation.

If the underlying system has a geometry that is more complicated than 1D, it stands to reason that constructing its state preparation via 1D methods will be inherently inefficient. There is therefore a pressing need to develop `geometry-native' state preparation methods that are tailored directly to the underlying entanglement geometry of the quantum state to be prepared. There is also great opportunity in tailoring the circuit connectivity of state-preparation routines to the native connectivity of the quantum processor, which in mature architectures nowadays is seldom 1D --- superconducting processors have fixed planar connectivity~\cite{kim2023evidence,google2025quantum}, trapped-ion processors can support all-to-all connectivity within a register~\cite{chen2024benchmarking,ransford202698}, neutral-atom arrays have reconfigurable connectivity by physically transporting atoms in optical tweezers~\cite{bluvstein2024logical}, and silicon-spin qubits support long-ranged connectivity by coherently shuttling electrons~\cite{cai2023looped,de2025high,siegel2026quantum,undseth2026weight,matsumoto2026two}.

\begin{figure*}
    \centering
    \includegraphics[width=\linewidth]{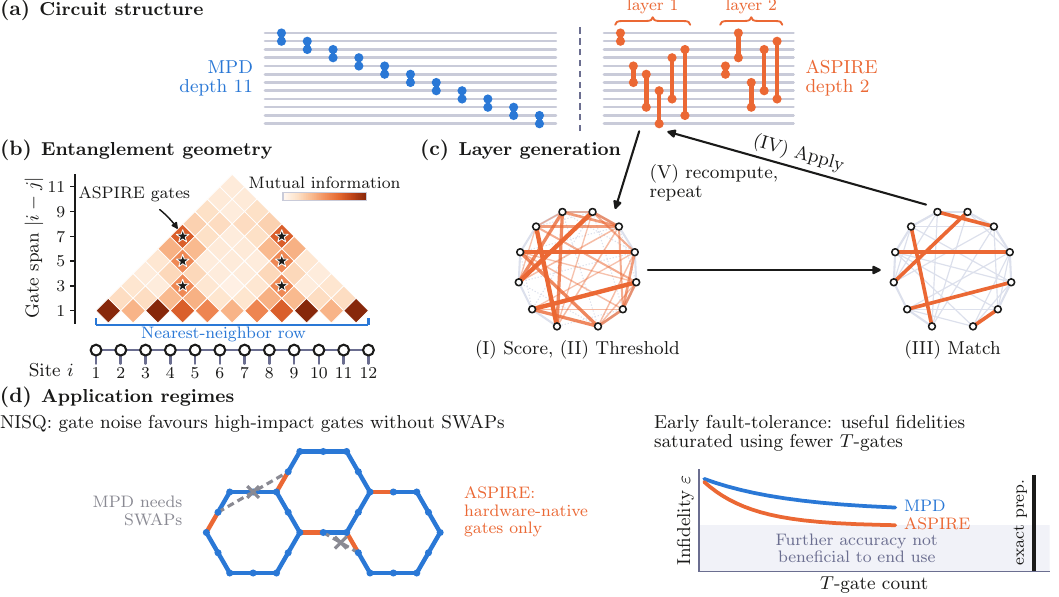}
    \caption{\textbf{Overview of the ASPIRE protocol.}
    \textbf{(a)} Comparison of circuit structure for the MPD~\cite{ran2020encoding} (blue) and ASPIRE (orange) at identical budget. The MPD produces sequential `staircase' layers, whereas ASPIRE selects pairs optimised to reduce long-ranged entanglement, parallelised at shallower depth.
    \textbf{(b)} Pairwise quantum mutual information Equation~\ref{eqn:quantum_mutual_information} for the ground state of a $4\times3$ Heisenberg model Equation~\ref{eq:heisen}, snake-mapped to an MPS and used to select pairs. Each pair scored is a cell at its MPS site (horizontal) and gate range (vertical); only the bottom row is accessible by standard MPD methods, while the six starred cells are connections only targetable by ASPIRE at ranges $3$, $5$ and $7$ in this example.
    \textbf{(c)} One layer of Algorithm~\ref{alg:state-prep}: score and threshold candidate pairs (I--II), take a maximum-weight disjoint set (III), apply it as a layer (IV), and recompute and repeat on the updated state (V).
    \textbf{(d)} The two regimes targeted by ASPIRE. On NISQ hardware with non-1D fixed connectivity (left), the
    MPS snake mapping needs SWAP networks to realise the MPD’s 1D nearest-neighbour gates (blue), with the jumps shown in gray crosses. In contrast, ASPIRE can be constrained to select only from gates that are native to the hardware (orange). In early fault-tolerant systems (right), ASPIRE (orange) can achieve the useful fidelity for a particular application at a fraction of the $T$-gate count for exact preparation, at an efficiency outperforming the MPD (blue).}
    \label{fig:overview}
\end{figure*}

Our work is a response to both of these geometric needs, enabling state preparation to be tailored to both the state's own geometry as well as the geometry imposed by the quantum hardware's connectivity. We introduce Adaptive State Preparation by Iterative Removal of Entanglement (ASPIRE), a protocol that improves upon prior MPD techniques by choosing long-ranged gate placements and configurations that directly remove entanglement, selected purely from these two notions of geometry and at full connectivity independent of the geometry of the state's representation. ASPIRE is a greedy classical algorithm, and builds an approximate state-preparation circuit from the state's entanglement geometry. Informally, our approach is summarised as follows (with a more detailed accounting given in Section~\ref{sec:protocol}, and a formal statement of the algorithm in Appendix~\ref{app:disentangling}):
\begin{informalalg} [ASPIRE (Informal)] (I) For every candidate move admitted by the connectivity matrix, evaluate a cost-normalised entanglement metric score.  (II) Discard moves below a certain threshold. (III) Select the maximum-weight set of pairwise disjoint moves. (IV) Apply the corresponding gates, forming one computational layer. (V) Recompute metrics, and repeat until no candidate survives step (II) or the native gate count reaches the budget $B$.
\label{alg:state-prep}
\end{informalalg}

The overall approach of our work is summarised in Figure~\ref{fig:overview}. The rest of the paper is organised as follows. In Section ~\ref{sec:methods} we cover the tensor-network preliminaries needed for defining the matrix product disentangler algorithm and contrast this with our entanglement-informed gate placement, before discussing different classical optimisation protocols. Next, Section ~\ref{sec:results} demonstrates the efficacy of ASPIRE in numerical experiments preparing Hamiltonians with long-range interactions, the addition of classical optimisation, the effect of increased hardware connectivity, or of adding routing costs and gate noise, finally testing also on multivariate functions. We discuss and conclude in Section ~\ref{sec:conclusion}.

\section{Methods}
\label{sec:methods}

\subsection{Tensor-network preliminaries}
In this work, we consider state preparation in the $N$-qubit register of a quantum computer, whose state is described by a state vector $\ket{\psi} \in \mathcal{H}^{\otimes N}$, where $\mathcal{H} \cong \mathbb{C}^2$ is the Hilbert space of a single qubit. A tensor network represents a multilinear map (most commonly a state or operator on this Hilbert space or a subspace) as a contraction of several lower-rank tensors, the most common being the MPS
\begin{equation}
\Psi_{\bm{s}}=[A_1]^{s_1}_{\alpha_1}[A_2]^{s_2}_{\alpha_1\alpha_2}\dots [A_{N-1}]^{s_{N-1}}_{\alpha_{N-2}\alpha_{N-1}} [A_N]^{s_N}_{\alpha_{N-1}},\label{eq:mps}
\end{equation}
representing the amplitudes $\Psi_{\bm{s}} = \bra{s_1 s_2 \dots s_N} \ket{\psi}$ of the state $\ket{\psi}$, and where the Einstein summation convention over repeated internal indices $\alpha_n$ is assumed. Tensor networks are most conveniently expressed in Penrose graphical notation~\cite{penrose1971applications,bridgeman2017hand}, where Equation~\ref{eq:mps} takes the form
\begin{equation}\Psi_{\bm{s}}=
\begin{array}{c}
\begin{tikzpicture}
\Vertex[x=0,y=0,label=$A_1$,position=above,size=0.25,color=blue!30]{A1}
\Vertex[x=1,y=0,label=$A_2$,position=above,size=0.25,color=blue!30]{A2}
\Vertex[x=2.5,y=0,label=$A_{N-1}$,position=above,size=0.25,color=blue!30]{Anm1}
\Vertex[x=3.5,y=0,label=$A_N$,position=above,size=0.25,color=blue!30]{An}
\draw[line width = 1] (A1) -- ++(0,-0.35) node[below] {\scriptsize$s_1$};
\draw[line width = 1] (A2) -- ++(0,-0.35) node[below] {\scriptsize$s_2$};
\draw[line width = 1] (Anm1) -- ++(0,-0.35) node[below] {\scriptsize$s_{N-1}$};
\draw[line width = 1] (An) -- ++(0,-0.35) node[below] {\scriptsize$s_N$};
\Edge[lw=1,position={below},label={\makebox[1em][c]{$\alpha_{N-1}$}}](Anm1)(An)
\Edge[lw=1,position={below},label={\makebox[1em][c]{$\alpha_1$}}](A1)(A2)
\Edge[lw=1,label=$\cdots$](A2)(Anm1)
\end{tikzpicture}
\end{array}.
\label{eqn:mps_diagram_form}
\end{equation}
The gauge freedom of tensor networks~\cite{schollwock2011density,bridgeman2017hand} means that we are free to place the tensor network in right-canonical form, where each tensor $A_n$ is an isometry such that
\begin{equation}
\begin{array}{c}
\begin{tikzpicture}
\Vertex[x=0,y=0.3,label=$A$,position=center,size=0.4,color=blue!30]{A1}
\Vertex[x=0,y=-0.3,label=$A^\dagger$,position=center,size=0.4,color=red!30]{A2}
\draw[line width = 1] (A1) -- ++(-0.4,0) node[below] {};
\draw[line width = 1] (A2) -- ++(-0.4,0) node[below] {};
\Edge[lw=1,bend=90](A1)(A2)
\Edge[lw=1](A1)(A2)
\end{tikzpicture}
\end{array}
=\begin{array}{c}
\begin{tikzpicture}
\draw[line width=1] (0,0.3) to[out=0,in=0,looseness=1.5] (0,-0.3);
\end{tikzpicture}
\end{array},
\end{equation}
i.e. $AA^\dagger = I$. An isometry is a physical operation that can be performed on a quantum computer, in this case representing an entangling operation between uninitialised and arbitrarily-initialised qubits. Concretely, it can be realised by a unitary whose columns are those of the isometry augmented by any orthonormal vectors spanning its nullspace. Operators $O$ are likewise representable by matrix product operators (MPOs) of the form
\begin{equation}\bra{\bm{s}'}O\ket{\bm{s}}=
\begin{array}{c}
\begin{tikzpicture}
\Vertex[x=0,y=0,position=above,size=0.25,color=blue!30]{A1}
\Vertex[x=1,y=0,position=above,size=0.25,color=blue!30]{A2}
\Vertex[x=2.5,y=0,position=above,size=0.25,color=blue!30]{Anm1}
\Vertex[x=3.5,y=0,position=above,size=0.25,color=blue!30]{An}
\draw[line width = 1] (A1) -- ++(0,-0.3) node[right] {\scriptsize$s'_1$};
\draw[line width = 1] (A2) -- ++(0,-0.3) node[right] {\scriptsize$s'_2$};
\draw[line width = 1] (Anm1) -- ++(0,-0.3) node[right] {\scriptsize$s'_{N-1}$};
\draw[line width = 1] (An) -- ++(0,-0.3) node[right] {\scriptsize$s'_N$};
\draw[line width = 1] (A1) -- ++(0,0.3) node[right] {\scriptsize$s_1$};
\draw[line width = 1] (A2) -- ++(0,0.3) node[right] {\scriptsize$s_2$};
\draw[line width = 1] (Anm1) -- ++(0,0.3) node[right] {\scriptsize$s_{N-1}$};
\draw[line width = 1] (An) -- ++(0,0.3) node[right] {\scriptsize$s_N$};
\Edge[lw=1,position={below},label={\makebox[1em][c]{}}](Anm1)(An)
\Edge[lw=1,position={below},label={\makebox[1em][c]{}}](A1)(A2)
\Edge[lw=1,label=$\cdots$](A2)(Anm1)
\end{tikzpicture}
\end{array},
\label{eqn:MPO_diagram_form}
\end{equation}
some of which may admit a low bond dimension, e.g. local Hamiltonian terms~\cite{hubig2017generic}, the quantum Fourier transform~\cite{chen2023quantum}, and derivative operators~\cite{garcia2021quantum,gourianov2022quantum,lee2026quantum}.

Applying gates to an MPS generally increases its bond dimension, so we often must compress it back to a lower target bond dimension $\chi$, i.e.
\begin{equation}
  \textsc{Compress}(\ket{\psi},\chi) \approx \operatorname*{argmin}_{|\tilde{\psi}\rangle\in\mathcal{M}_\chi} \left\|\ket{\psi}-|\tilde{\psi}\rangle\right\|,
  \label{eq:compress}
\end{equation}
where $\mathcal{M}_\chi$ denotes the set of normalised MPS with bond dimension $\le\chi$. This can be achieved by iteratively performing singular value decompositions along the chain and truncating the resultant Schmidt spectrum to retain no more than $\chi$ singular values~\cite{schollwock2011density}, or by more sophisticated methods~\cite{fishman2022itensor,harrow2025randomized}, and can also be done in a controlled manner up to some error tolerance depending on the method used.

\subsection{Matrix product disentangler}
Matrix product disentangler (MPD) techniques~\cite{ran2020encoding,zhou2021automatically,rudolph2023decomposition,ben2024approximate,sims2026multivariate,green2025quantum,bohun2026entanglement,szoldra2026scalable,mansuroglu2026preparation,sugawara2025embedding,szoldra2026belief} have largely built upon the iterative disentangling framework introduced by Ran~\cite{ran2020encoding}, which we briefly review here. Preparing an MPS of general bond dimension $\chi$ is a non-trivial procedure that typically requires compiling isometries/unitaries acting simultaneously on at least $\lceil\log_2(\chi)\rceil+1$ qubits~\cite{schoen2005sequential,iten2016quantum,iten2019introduction,malvetti2021quantum,fomichev2024initial,berry2025rapid,rupprecht2026sparse}; however, an MPS of bond dimension $\chi=2$ can be prepared \emph{exactly} with a single-layer `ladder' of $N-1$ two-qubit unitaries acting sequentially along the chain. The intuition behind the MPD approach of Ref.~\cite{ran2020encoding} is that truncating the MPS to $\chi=2$ and constructing this layer preserves some of the entanglement structure of the original MPS; this can be used as a primitive component to iteratively \emph{disentangle} a target MPS. In this approach, one truncates the MPS to $\chi=2$, constructs the sequential layer that prepares this state, applies its inverse to the untruncated MPS, and repeats (up to some allowable number of layers or fidelity target), which typically converges towards a circuit that approximately disentangles the MPS to a product state. The inverse of this disentangling circuit then serves as an approximate state preparation circuit for the MPS. This approach is formalised in Algorithm~\ref{alg:mpd}, for $K$ layers, fidelity threshold $f_t$, and truncation bond dimension $\chi_w$.

\begin{figure}\color{black}\algrenewcommand\alglinenumber[1]{{\color{black}\footnotesize#1:}}
\begin{algorithm}[H]
\caption{\color{black}Matrix product disentangler~\cite{ran2020encoding}}\label{alg:mpd}
\begin{algorithmic}[1]
\Procedure{MPD}{$\ket{\psi},K,f_{\mathrm{t}},\chi_w$}
    \State $\ket{\phi_0}\gets\ket{\psi}$
    \For{$k=1,\dots,K$}
        \State $\ket{\tilde\phi}\gets$ \Call{Compress}{$\ket{\phi_{k-1}},2$} \Comment{right-canonical, Eq.~\ref{eq:compress}}
        \For{$n=1,\dots,N-1$}
            \State $G_n\gets$ \Call{CompleteUnitary}{$A_n$} \Comment{Eq.~\ref{eqn:completeunitary}}
        \EndFor
        \State $G_{N-1}\gets(I\otimes A_N)\,G_{N-1}$ \Comment{absorb single-qubit unitary on qubit $N$}
        \State $L_k\gets G_{N-1}\cdots G_2\,G_1$ \Comment{$G_1$ acts first}
        \State $\ket{\phi_k}\gets$ \Call{Compress}{$L_k^\dagger\ket{\phi_{k-1}},\chi_w$}
        \If{$|\langle 0^{\otimes N}|\phi_k\rangle|\ge f_{\mathrm{t}}$}
            \State \textbf{break}
        \EndIf
    \EndFor
    \State \Return $C=L_1L_2\cdots L_k$
\EndProcedure
\end{algorithmic}
\end{algorithm}
\end{figure}
To construct the preparation circuit for the $\chi=2$ MPS, we must map the isometries composing the MPS to unitary gates that act on the corresponding qubits.
If the truncated MPS is not already in right-canonical form from the compression procedure, it is first transformed to this gauge.
The leftmost tensor (orthogonality centre) is then a 2-qubit state-preparation operation ($4\times 1$ isometry), the bulk tensors are $4\times2$ isometries, and the rightmost tensor is a $2\times2$ single-qubit unitary. Each isometry $A_n$ for $n=1,\dots,N-1$ is augmented to a unitary $G_n\equiv \textsc{CompleteUnitary}(A_n)$ satisfying
\begin{equation}
\begin{array}{c}
\begin{tikzpicture}
\Vertex[x=0,y=0,shape=rectangle,label=$G_n$,position=center,color=blue!15,style={minimum width=1cm,minimum height=0.55cm,rounded corners=2pt}]{G}
\draw[line width = 1] ([xshift=-0.35cm]G.north) -- ++(0,0.2) node[left] {\scriptsize$\alpha_{n-1}$};
\draw[line width = 1] ([xshift=0.35cm]G.north) -- ++(0,0.2) node[right] {\scriptsize$\ket{0}$};
\draw[line width = 1] ([xshift=-0.35cm]G.south) -- ++(0,-0.2) node[left] {\scriptsize$s_n$};
\draw[line width = 1] ([xshift=0.35cm]G.south) -- ++(0,-0.2) node[right] {\scriptsize$\alpha_n$};
\end{tikzpicture}
\end{array}
=
\begin{array}{c}
\begin{tikzpicture}
\Vertex[x=0,y=0,label=$A_n$,position=above,size=0.25,color=blue!30]{A}
\draw[line width = 1] (A) -- ++(-0.3,0) node[left] {\scriptsize$\alpha_{n-1}$};
\draw[line width = 1] (A) -- ++(0.3,0) node[right] {\scriptsize$\alpha_n$};
\draw[line width = 1] (A) -- ++(0,-0.3) node[below] {\scriptsize$s_n$};
\end{tikzpicture}
\end{array},
\label{eqn:completeunitary}
\end{equation}
i.e.\ $G_n$ maps the bond index on qubit $n$ to the physical index on qubit $n$ and passes the next bond index to qubit $n+1$ to be processed by the next gate (for $n=1$ both inputs are $\ket{0}$). This is straightforwardly achieved by augmenting $A_n$ with orthonormal columns that span the nullspace. The single-qubit unitary on qubit $N$ can be absorbed into the preceding two-qubit gate without additional cost, such that each layer consists of $N-1$ general 2-qubit unitaries $G_n\in\operatorname{SU}(4)$. The newest layer $L_k$ acts first on $\ket{0}^{\otimes N}$ (last step of disentanglement), which is why $C=L_1L_2\cdots L_k$ is returned in this order.

While the sequential MPD has been a useful starting point for approximate state preparation routines, particularly for states with short-ranged 1D correlations, sequential circuits spread entanglement slowly: acting only on nearest-neighbour pairs in the 1D MPS chain, they extend the correlated region by roughly one site per sweep --- the slow growth of this lightcone makes these circuits prohibitively deep for states with non-1D or long-ranged correlations~\cite{robertson2025approximate,jaderberg2026variational,pennington2026preparing}. This limitation is one of several motivations for the approach developed in this work.

\subsection{Entanglement-informed gate placement} \label{sec:protocol}
While brick-wall and parallel-sequential ans\"{a}tze circumvent some of the limitations of sequential circuits by parallelising entangling gates~\cite{lin2021real,haghshenas2022variational,robertson2025approximate,jaderberg2026variational,szoldra2026scalable,pennington2026preparing,wei2026state}, they still remain geometrically local on the 1D ordering --- this incurs a depth that grows with the correlation length, precludes tailoring to the connectivity of the quantum processor, and fails to accommodate the true entanglement geometry of non-1D systems represented with MPS. We circumvent this here by introducing a protocol that places two-qubit gates on arbitrary pairs of sites according to an entanglement metric, allowing long-ranged correlations to be directly captured rather than slowly propagated site-by-site over many layers.

Directly targeting the reduction of local entanglement (as opposed to the truncation heuristic of original approaches~\cite{ran2020encoding}) has proven to be an effective strategy for enhancing the MPD~\cite{green2025quantum,mansuroglu2026preparation}. However, existing methods have targeted the local entanglement at the level of individual bonds, which can be a suboptimal strategy, particularly for the aforementioned states with long-ranged correlations where MPS methods still remain performant~\cite{chan2011density,baiardi2020density,menczer2024tensor,menczer2024parallel,xiang2024distributed,legeza2025orbital,white2007neel,yan2011spin,stoudenmire2012studying,zheng2017stripe,zhou2020limits,garcia2021quantum,ye2022quantum,gourianov2022quantum,shinaoka2023multiscale,ritter2024quantics,nunez2025learning,bohun2026entanglement,sims2026multivariate,lee2026quantum,connor2026tensor,niedermeier2026solving}. Iteratively selecting pairs of sites with high quantum mutual information (QMI) and removing entanglement on these pairs has been an effective means to discover the entanglement geometry of quantum states with complicated correlations, recovering scale-invariant and holographic AdS-like structure in relevant quantum states~\cite{hyatt2017extracting}. Our protocol applies similar ideas constructively to design state-preparation circuits tailored to the underlying entanglement structure of the target state.

\subsubsection{Entanglement metrics}

Freeing our circuit construction to allow arbitrary (or architecture-constrained) pairings of qubits allows for drastically more efficient generation and spreading of entanglement and circuit expressivity, but also introduces the challenge of navigating the accordingly larger search space of possible gate placements. Suitable entanglement metrics are therefore essential to guide placement; we will frequently use the two-site QMI as our guiding metric, which depends on the reduced density matrix of the two qubits $i,j\in\{1,\dots,N\}$:
\begin{equation}
    \rho_{ij}=\Tr_{\overline{ij}}[\ket{\psi}\bra{\psi}],\label{eqn:reduced_density_matrix}
\end{equation}
where $\overline{ij}\equiv\{1,\dots,N\}\setminus\{i,j\}$ denotes the set of all qubits except $i$ and $j$ (hence the partial trace is computed over all other qubits). The quantum mutual information is then given by
\begin{equation}
  I(i:j) \equiv S(\rho_i) + S(\rho_j) - S(\rho_{ij}),\label{eqn:quantum_mutual_information}
\end{equation}
where $S(\rho) \equiv -\operatorname{Tr} \rho \log \rho$ is the von Neumann entropy with $\log\equiv\log_2$ throughout, so that all entropies and correlations are quoted in bits \cite{nielsen2010quantum} and $\rho_i = \Tr_{j}[\rho_{ij}]$ is the reduced density matrix for qubit $i$ (likewise for $j$).

When the two-site QMI is required, we compute $\rho_{ij}$ via tensor-network contractions as 
\begin{equation}
\bra{s_i',s_j'}\rho_{ij}\ket{s_i,s_j}=
\begin{array}{c}
\begin{tikzpicture}
% ket (top) row: 1, i, i+1, j-1, j, N
\Vertex[x=0,y=-0.2,size=0.25,color=blue!30]{K1}
\Vertex[x=0.9,y=-0.2,size=0.25,color=blue!30]{Ki}
\Vertex[x=1.4,y=-0.2,size=0.25,color=blue!30]{Kii}
\Vertex[x=2.3,y=-0.2,size=0.25,color=blue!30]{Kjm}
\Vertex[x=2.8,y=-0.2,size=0.25,color=blue!30]{Kj}
\Vertex[x=3.7,y=-0.2,size=0.25,color=blue!30]{KN}
% bra (bottom) row
\Vertex[x=0,y=0.2,size=0.25,color=red!30]{B1}
\Vertex[x=0.9,y=0.2,size=0.25,color=red!30]{Bi}
\Vertex[x=1.4,y=0.2,size=0.25,color=red!30]{Bii}
\Vertex[x=2.3,y=0.2,size=0.25,color=red!30]{Bjm}
\Vertex[x=2.8,y=0.2,size=0.25,color=red!30]{Bj}
\Vertex[x=3.7,y=0.2,size=0.25,color=red!30]{BN}
% traced pairs: closed ket-bra loops (straight + bent \Edge)
%\Edge[lw=1](K1)(B1)
\Edge[lw=1,bend=-135](K1)(B1)
%\Edge[lw=1](Kii)(Bii)
\Edge[lw=1,bend=-135](Kii)(Bii)
%\Edge[lw=1](Kjm)(Bjm)
\Edge[lw=1,bend=-135](Kjm)(Bjm)
%\Edge[lw=1](KN)(BN)
\Edge[lw=1,bend=-135](KN)(BN)
%open-leg lines
\draw[line width=1] (Ki) -- ++(0,-0.3) node[right] {\scriptsize$s_i'$};
\draw[line width=1] (Bi) -- ++(0,0.3) node[right] {\scriptsize$s_i$};
\draw[line width=1] (Kj) -- ++(0,-0.3) node[right] {\scriptsize$s_j'$};
\draw[line width=1] (Bj) -- ++(0,0.3) node[right] {\scriptsize$s_j$};
% ket-row bonds (elided gaps use \cdots)
\Edge[lw=1,label=$\cdots$](K1)(Ki)
\Edge[lw=1](Ki)(Kii)
\Edge[lw=1,label=$\cdots$](Kii)(Kjm)
\Edge[lw=1](Kjm)(Kj)
\Edge[lw=1,label=$\cdots$](Kj)(KN)
% bra-row bonds
\Edge[lw=1,label=$\cdots$](B1)(Bi)
\Edge[lw=1](Bi)(Bii)
\Edge[lw=1,label=$\cdots$](Bii)(Bjm)
\Edge[lw=1](Bjm)(Bj)
\Edge[lw=1,label=$\cdots$](Bj)(BN)
\end{tikzpicture}
\end{array}
\label{eqn:rdm_diagram_form}
\end{equation}
where the red nodes denote the conjugated (bra) factor in the projector $\ket{\psi}\bra{\psi}$. This network can be contracted efficiently, as noted in Ref.~\cite{schollwock2011density}.

Following the approach in Ref.~\cite{hyatt2017extracting}, when a pair of sites $(i,j)$ has been selected, it is straightforward to construct the unitary operation $U_{ij}\in\operatorname{SU}(4)$ acting on these sites that maximally reduces the QMI evaluated on the rotated state by solving
\begin{equation}
  U^*_{ij} \equiv \operatorname{argmin}_{U_{ij}} I(U_{ij} \rho_{ij} U_{ij}^\dagger).\label{eqn:optimal_entanglement_removing_unitary}
\end{equation}
Choosing each gate to minimise an entanglement measure, rather than to maximise its overlap with a truncated target, is also the principle of the classical variational disentangler of Ref.~\cite{mansuroglu2026preparation}, which optimises a layer of parameterised disentangling gates by minimising bipartite entanglement measures after its application. Equation~\ref{eqn:optimal_entanglement_removing_unitary} differs in acting on an arbitrary pair of sites and in being solved exactly rather than variationally (Appendix~\ref{app:gate}).
Without searching through $SU(4)$, the optimal gate maps the eigenvectors of $\rho_{ij}$, ordered by decreasing eigenvalue, onto $\ket{00},\ket{01},\ket{10},\ket{11}$ (Proposition~\ref{prop:pair}, \cite{jevtic2012maximally, jevtic2012quantum}), so a single $4\times4$ eigendecomposition yields both the gate and the correlation it removes. This allows us to define the \emph{removable} QMI:
\begin{definition}[Removable quantum mutual information (RQMI)]
  Let $\rho_{ij}$ be the reduced density matrix of sites $(i,j)$ as defined in Equation~\ref{eqn:reduced_density_matrix} and $U^*_{ij}$ be as defined in Equation~\ref{eqn:optimal_entanglement_removing_unitary}. The removable quantum mutual information (RQMI) is then defined as
  \begin{equation}
    R_{ij}=S(\rho_i) + S(\rho_j) - S((\rho')_i) - S((\rho')_j)
    \label{eq:RQMI}
  \end{equation}
  where $\rho'\equiv U^*_{ij} \rho_{ij} (U^*_{ij})^\dagger$.
\end{definition}
For more details regarding the protocol, see Appendix~\ref{app:disentangling}.

\subsubsection{Placing gates}
\label{sec:placement}

For every pair of MPS sites $(i,j)$, $R_{ij}$ quantifies the correlation that the optimal 2-qubit rotation acting on those sites could remove. In unrestricted ASPIRE, all such pairwise scores are computed and a layer is generated by taking a maximum-weight matching of the candidate pairs, removing as much correlation as possible in depth-1. Adding restrictions to the protocol allows it to be tailored towards different noise regimes, hardware connectivities, and optimisation protocols.

The admissible gate locations are given by a symmetric connectivity matrix $K\in\{0,1\}^{N\times N}$, with $K_{ij}=1$ if a gate is permitted on qubits $(i,j)$. All-to-all connectivity sets $K_{ij}=1$ for every $i\neq j$. The connectivity radius $r$ used throughout Section~\ref{sec:results} admits only pairs with $|i-j|\le r$ in MPS site order, so that $r=1$ recovers the nearest-neighbour pairs available to the MPD. On a device with fixed connectivity, such as the heavy-hex patch of Section~\ref{subsec:noise}, $K$ may be the device's coupling graph, or it may also admit pairs the device can couple only through SWAP gates, in which case those gates carry a routing cost.

Each candidate $(i,j)$ is scored by its removable QMI $R_{ij}$, obtained together with its gate $U^*_{ij}$ from one eigendecomposition of $\rho_{ij}$ (Appendix~\ref{app:gate}). When candidates differ in cost, as they do when some must be routed, we score instead by the correlation removed per native two-qubit gate, $\tilde R_{ij}=R_{ij}/c_{ij}$. Here $c_{ij}$ counts the native gates the move costs, including its routing: a generic two-qubit gate compiles to three CNOTs, and every SWAP needed to bring qubits $i$ and $j$ together adds three more. The normalised score is the appropriate one for connections the hardware must route, and the raw score for connections it supports natively (Appendix~\ref{app:connectivity}).

Candidates scoring below a threshold $\tilde\tau$, measured in bits per native gate, are discarded. The threshold keeps the protocol from spending gates on pairs whose removable correlation is negligible. Similarly candidates scoring below a relative threshold, for example not scoring in the top decile of possible pairs, can also be ignored for similar gate-frugal reasons. Hybrid approaches can utilise both restrictions, including having it vary throughout the protocol (for example loosening the threshold after a number of layers have been applied). These limitations can also end the circuit construction once no candidate survives. The construction also ends when the native gate count reaches the budget $B$.

Gates on disjoint pairs of qubits commute and can be executed simultaneously, so each circuit layer is built from pairwise disjoint candidates. Disjointness also makes the scores additive. A two-qubit gate changes only the marginals of the two qubits it acts on, so the correlation removed by a layer of disjoint gates is exactly the sum of what each gate removes, with no interaction between them (Proposition~\ref{prop:descent}). Selecting a layer is then the problem of choosing, among the surviving candidates, the disjoint set of pairs with the largest total score. If the qubits are the vertices of a graph and the surviving candidates its edges, weighted by their scores, this is a maximum-weight matching, which Edmonds' blossom algorithm solves exactly in $\mathcal{O}(N^3)$ time~\cite{edmonds1965paths}. For the largest systems we restrict the matching to the $32$ highest-scoring candidates, over which it is still found exactly. 

A separate parameter determines how many gates a layer admits. A full matching commits up to $N/2$ gates on the basis of a single scoring of the state, whereas smaller layers re-score the state more often within the same gate budget. Neither is better in general. At equal gate count, narrower layers reached a lower infidelity on the Heisenberg ground states and random MPS we tested, and the full matching did on the TFIM, so the layer size has to be chosen together with the gate budget (Appendix~\ref{app:layer}). Unless stated otherwise, a layer is the full matching. The connectivity study of Section~\ref{subsec:connectivity} uses a simpler greedy rule instead. Candidates are taken in descending order of their mutual information $I(i{:}j)$, and a pair is admitted only if it is disjoint from the pairs already chosen and each of its qubits is weakly correlated with theirs, $I(a{:}b)\le\kappa=0.1$~bits for every such cross pair $(a,b)$, so that each layer acts on nearly independent parts of the state.

The selected gates $U^*_{ij}$ are applied to the MPS, which, as in the MPD, removes their correlation from the working state. Gates on pairs that are not adjacent in the chain are applied through MPS SWAP chains, while truncating the bond dimension. The two-site reduced density matrices are then recomputed on the updated state, and the next layer is selected. Recomputation is necessary because a gate on $(i,j)$ also changes $\rho_{ik}$ and $\rho_{jk}$ for every other qubit $k$, so a pair that scored poorly before a layer may score highly after it. Grouping the candidate pairs by their left site, the scan over all $\mathcal{O}(N^2)$ candidates costs $\mathcal{O}(N^2\chi^3)$ per layer. When the construction stops, the inverses of the applied gates in reverse order form the state-preparation circuit, exactly as in Algorithm~\ref{alg:mpd}. Figure~\ref{fig:overview}(c) traces one pass of this loop.

This raises the question of why removing correlations is beneficial for state preparation. The circuit succeeds when the disentangled state $\ket{\phi}=C^\dagger\ket{\psi_{\text{target}}}$ is close to a product state, because the free layer of single-qubit rotations then maps it onto $\ket{0}^{\otimes N}$. A natural measure of that distance is the total correlation $T=\sum_a S(\rho^a)$, the sum of the single-site entropies of $\ket{\phi}$, which vanishes exactly when $\ket{\phi}$ is a product. We show below that $T$ is the relative-entropy distance of $\ket{\phi}$ to the product states and that a layer of disjoint gates lowers it by exactly the sum of what each gate removes, with $R_{ij}$ the largest decrease any gate on $(i,j)$ can achieve, so placing gates by removable QMI is greedy descent on this distance. The same single-site marginals also bound the infidelity of the circuit (Theorem~\ref{thm:cert}).

For a general support $A\subseteq\{1,\dots,N\}$ we can use the total correlation to state Proposition~\ref{prop:descent}. Writing $S(\rho\Vert\sigma)=\operatorname{Tr}[\rho(\log\rho-\log\sigma)]$ for the relative entropy between density matrices $\rho$ and $\sigma$, for pure states $\ket\psi$ one has $\min_{\sigma=\bigotimes_{a\in A}\sigma_a}S(\dyad{\psi}\Vert\sigma)=S\bigl(\dyad{\psi}\Vert\bigotimes_{a\in A}\rho^a\bigr)=\sum_{a\in A}S(\rho^a)=T_a$, where each $\sigma_a$ is an arbitrary state on site $a$, so Algorithm~\ref{alg:state-prep} descends on a canonical distance rather than on a heuristic score. 

\begin{proposition}[Exact descent and exact additivity]
\label{prop:descent}
Let $\ket{\phi}$ be pure and let $\mathcal L=\{A_g\}$, using $g$ as the gate-index, be pairwise disjoint
supports carrying $U_{\mathcal L}=\bigotimes_g U_{A_g}$. Then
\begin{equation}
  \Delta T=\sum_g \Delta_g,
  \ \
  \Delta_g=\sum_{a\in A _g}\bigl[S((\rho')^{a})-S(\rho^{a})\bigr] \ge -R_{A_g},
  \label{eq:descent}
\end{equation}
with $R_{A_g}$ evaluated on the pre-layer state and equality if and only if
$U_{A_g}$ attains the minimum in:
\begin{equation}
    R_A=\sum_{a\in A}S(\rho^{a})-\min_{U_A}\sum_{a\in A 
    }S\bigl((\rho')^{a}\bigr).
\end{equation}
Sites outside
$\bigcup_g A_g$ do not contribute, and there is no cross term between supports
at any order.
\end{proposition}

Proposition~\ref{prop:descent} identifies Algorithm~\ref{alg:state-prep} with raw scores and no thresholds as greedy coordinate descent on $T$ constructed layer-wise and without cross terms or single-qubit rotations affecting $T$. The main impact of our protocol is therefore the choice of which two-site coordinate block to descend on. The minimisation defining $R_A$ is not a continuous search over $\mathrm{U}(d^{|A|})$ but a finite assignment over the eigenbasis of $\rho^A$, so one eigendecomposition yields the score and the gate together; see Appendix~\ref{app:gate}.

 What ultimately matters is the infidelity $\varepsilon_\star$ of the prepared state, defined in Equation~\ref{eq:infidelity-frame}. Theorem~\ref{thm:cert} bounds it by the weight that the single-site marginals of the disentangled state place outside their top eigenvectors.

\begin{theorem}[Infidelity bound from single-site marginals]
\label{thm:cert}
Let $\lambda_a$ be the largest eigenvalue of the $a$-th single-site marginal of
$\ket{\phi}=C^{\dagger}\ket{\psi_{\text{target}}}$, and let
$\Lambda=\sum_a(1-\lambda_a)$ and $\tilde\Lambda=\Lambda/(2-\Lambda)$. Then
\begin{equation}
  \varepsilon_{\star}\ \le\ 1-\sqrt{1-\tilde\Lambda}\ \le\ \tilde\Lambda\ \le\ \Lambda\ \le\ T/2 ,
  \label{eq:cert}
\end{equation}
the first three inequalities holding when $\Lambda\le1$ and the last
unconditionally, for sites of any local dimension. 
\end{theorem}
A proof is given in Appendix~\ref{app:thm2}.
For small $\Lambda$ the certificate reads $\varepsilon_\star\le\Lambda/4+O(\Lambda^{2})$, and the coefficient cannot be improved (Appendix~\ref{app:thm2}). Theorem~\ref{thm:cert} is a statement about the exact residual $C^{\dagger}\ket{\psi_{\text{target}}}$. Algorithm~\ref{alg:aspire} compresses its working state after every gate, and compression can remove correlation that the recorded circuit has not, so the marginals of the working state certify the circuit only once the accumulated truncation is included. Because the Fubini--Study angle is unitarily invariant and obeys the triangle inequality, $\arccos(1-\varepsilon_\star)\le\arcsin\sqrt{\tilde\Lambda_w}+\sum_k\delta_k$ whenever $\Lambda_w\le1$, where $\Lambda_w$ and $\tilde\Lambda_w$ are evaluated on the working state and $\delta_k$ is the Fubini--Study angle between the working state just before and just after its $k$-th compression. The correction vanishes when no truncation occurs. The infidelities we report do not rely on it, since they are computed separately by applying $C^{\dagger}$ to the target afresh, and Appendix~\ref{app:targets} checks that evaluation against a larger bond-dimension ceiling.

$T$ and $\Lambda$ are sums over sites of a function of a single-site marginal, so they are invariant under every permutation of site labels, whereas the bond entropies $S_n$ that the MPD descends on are defined by a contiguous bipartition and are not. Therefore, the site ordering can negatively impact the MPD, but not ASPIRE when afforded full all-to-all connectivity.

Before turning to the optimisation of the gates, we illustrate why placing gates by entanglement, rather than by position along the chain, pays off. We test ASPIRE on two toy models designed to highlight the impact of increased connectivity in the finite-layer regime and gate selection based on entanglement metrics, comparing it to the MPD in Figure~\ref{fig:toy-expressivity}. In \textbf{(a)} and \textbf{(c)} we consider the $N$-qubit GHZ state~\cite{greenberger1990bell} $\ket{\mathrm{GHZ}_N}\equiv(\ket{0}^{\otimes N}+\ket{1}^{\otimes N})/\sqrt{2}$, an MPS of bond dimension 2 whose pairwise correlations do not decay with distance and thus have range $N$. In \textbf{(b)}, we consider a $L_x \times L_y$ lattice carrying a Bell pair on each vertical bond snake-mapped to MPS form from side-to-side. In the first two panels \textbf{(a-b)} we demonstrate that ASPIRE with different gate ranges can decrease the total correlation faster than the MPD for both circuit depth and gate count. In \textbf{(a)} there is a rate of correlation removal per unit depth that no nearest-neighbour method can reach that ASPIRE breaches when its gates use a longer range. The bound behind the shaded region is proved in Appendix~\ref{app:ghzbound}. For \textbf{(b)} a range exceeding $2 L_x-1$ means that all the Bell-pair connections are targetable by the algorithm and the correlation can be removed in orders of magnitude fewer gates than the MPD, which cannot place gates along most of the Bell-pair connections due to the snake mapping. In \textbf{(c)} we demonstrate that all-to-all ASPIRE can exactly prepare the $\text{GHZ}_{N}$ chain in $\lceil \log_2 N \rceil$ layers, exponentially separated from the staircase's optimal $N-1$ layers \cite{malz2024preparation}. The all-to-all construction beats the staircase because unrestricted gate span lets the light cone grow exponentially in depth rather
than linearly, and is itself optimal as a disjoint layer of two-qubit gates can at most double the causal cone of any site, thus requiring a depth $D$ of $2^D \geq N$ (Proposition~\ref{prop:ghzbound}), and Appendix~\ref{app:ghz} records the schedule and the corresponding statement at bounded range.

\begin{figure*}
    \centering
    \includegraphics[width=\linewidth]{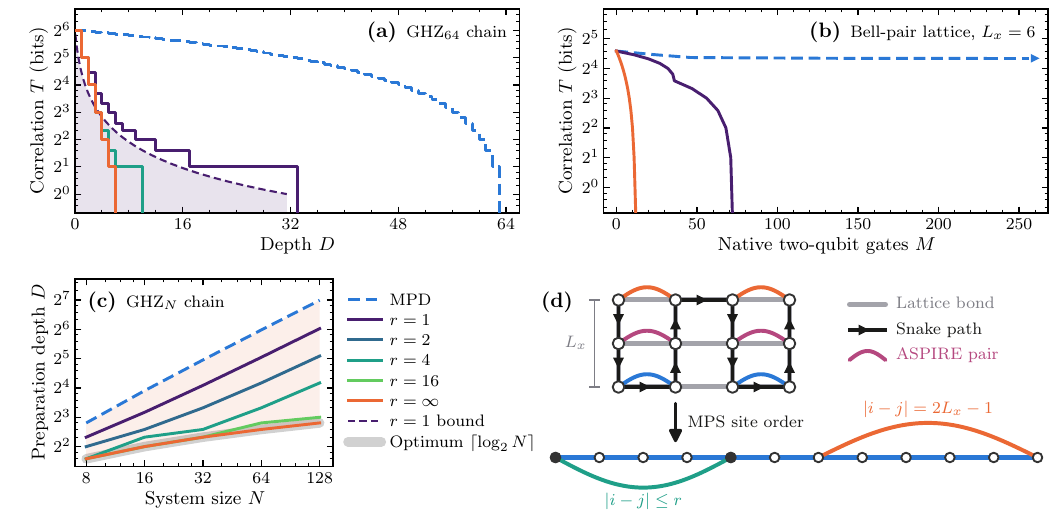}
    \caption{ \textbf{A toy model for the impact of better connectivity.} Top row: the remaining correlation $T$ is shown as a function of \textbf{(a)} circuit depth for a $\mathrm{GHZ}_{64}$ system and of \textbf{(b)} native gate count for an $L_x \times L_y = 6 \times 6$ lattice carrying a Bell pair on each vertical bond. For \textbf{(a)} the shaded region is unreachable by nearest-neighbour circuits, with the MPD's performance shown in dashed blue and ASPIRE at different connectivity radii shown as solid coloured lines. For \textbf{(b)} the Bell pair lattice's entanglement becomes long-ranged when snake mapped to an MPS, seeing all the Bell pairs only when the connectivity radius reaches $2L_x-1$. \textbf{(c)} Depth to exact preparation of the $\mathrm{GHZ}_{64}$, requiring $\lceil \log_2 N \rceil$ layers at unrestricted connectivity, and $N-1$ layers for the staircase construction.  Finite-radius curves in \textbf{(a)} and \textbf{(c)} follow the analytic schedule of Appendix~\ref{app:ghz}; the $r=\infty$ and Bell-lattice curves are ASPIRE output.
    Also shown in \textbf{(d)} is a schematic of a smaller lattice, turned so that the snake runs down its columns and the Bell pairs lie on its horizontal bonds, with the lattice bonds in grey, the snake path in black, and each Bell pair coloured by the chain distance $|i-j|$ it acquires in MPS site ordering, which is the connectivity radius $r$ needed to reach it.}
    \label{fig:toy-expressivity}
\end{figure*}

\subsection{Optimisation}\label{sec:optimisation}

Both standard MPD methods and ASPIRE determine circuit structure and gate placement. However, the gates themselves can be further refined by optimisation against a target objective, such as the infidelity of the prepared state $C\ket{0}^{\otimes N}$ with respect to the target state $\ket{\psi_{\text{target}}}$, given by
\begin{equation}
  \varepsilon = 1 - f, \qquad
  f =
  \big|\bra{\psi_{\mathrm{target}}} C \ket{0}^{\otimes N}\big| ,
  \label{eq:infidelity}
\end{equation}
where $C$ is the compiled state-preparation circuit. Throughout this work, this is the primary metric used to evaluate the quality of the compiled circuits.
Unless stated otherwise, $\ket{\psi_{\text{target}}}$ is the bond-dimension-$\chi$ MPS, so $\varepsilon$ measures the compilation error of the circuit and neglects the truncation error of the underlying tensor-network representation of the state --- this is a standard assumption~\cite{ran2020encoding,zhou2021automatically,rudolph2023decomposition,ben2024approximate,green2025quantum,mansuroglu2026preparation,szoldra2026scalable,jaderberg2026variational,pennington2026preparing,gibbs2026low,szoldra2026belief} that is rarely challenged~\cite{fomichev2024initial,berry2025rapid}, but we also describe a means to overcome it in Section~\ref{sec:energy_refinement}. $f$ can be computed efficiently with standard tensor-network gate application~\cite{vidal2003efficient} and MPS-MPS contraction~\cite{schollwock2011density}, e.g. for an $N=4$ system with a 5-gate brickwork circuit, the overlap inside the brackets in Equation~\ref{eq:infidelity} can be written diagrammatically as
\begin{equation}
\bra{\psi_{\mathrm{target}}}C\ket{0^{\otimes N}}=
\begin{array}{@{}c@{}}
\begin{tikzpicture}
% target MPS <psi| (left column)
\Vertex[x=0,y=0,size=0.2,color=red!30]{P1}
\Vertex[x=0,y=-0.4,size=0.2,color=red!30]{P2}
\Vertex[x=0,y=-0.8,size=0.2,color=red!30]{P3}
\Vertex[x=0,y=-1.2,size=0.2,color=red!30]{P4}
\Edge[lw=1](P1)(P2)
\Edge[lw=1](P2)(P3)
\Edge[lw=1](P3)(P4)
% brickwork C = U_5 U_4 U_3 U_2 U_1, with U_1, U_2 acting first on |0>
\Vertex[x=0.5,y=-0.2,shape=rectangle,label=$U_4$,position=center,fontsize=\scriptsize,color=blue!15,style={minimum width=0.42cm,minimum height=0.6cm}]{G4}
\Vertex[x=0.5,y=-1.0,shape=rectangle,label=$U_5$,position=center,fontsize=\scriptsize,color=blue!15,style={minimum width=0.42cm,minimum height=0.6cm}]{G5}
\Vertex[x=1.1,y=-0.6,shape=rectangle,label=$U_3$,position=center,fontsize=\scriptsize,color=blue!15,style={minimum width=0.42cm,minimum height=0.6cm}]{G3}
\Vertex[x=1.7,y=-0.2,shape=rectangle,label=$U_1$,position=center,fontsize=\scriptsize,color=blue!15,style={minimum width=0.42cm,minimum height=0.6cm}]{G1}
\Vertex[x=1.7,y=-1.0,shape=rectangle,label=$U_2$,position=center,fontsize=\scriptsize,color=blue!15,style={minimum width=0.42cm,minimum height=0.6cm}]{G2}
% wires: MPS -> layer nearest the MPS
\draw[line width = 1] (P1) -- ([yshift=0.2cm]G4.west);
\draw[line width = 1] (P2) -- ([yshift=-0.2cm]G4.west);
\draw[line width = 1] (P3) -- ([yshift=0.2cm]G5.west);
\draw[line width = 1] (P4) -- ([yshift=-0.2cm]G5.west);
% wires: between layers
\draw[line width = 1] ([yshift=0.2cm]G4.east) -- ([yshift=0.2cm]G1.west);
\draw[line width = 1] ([yshift=-0.2cm]G4.east) -- ([yshift=0.2cm]G3.west);
\draw[line width = 1] ([yshift=0.2cm]G5.east) -- ([yshift=-0.2cm]G3.west);
\draw[line width = 1] ([yshift=-0.2cm]G5.east) -- ([yshift=-0.2cm]G2.west);
\draw[line width = 1] ([yshift=0.2cm]G3.east) -- ([yshift=-0.2cm]G1.west);
\draw[line width = 1] ([yshift=-0.2cm]G3.east) -- ([yshift=0.2cm]G2.west);
% initial state |0>
\draw[line width = 1] ([yshift=0.2cm]G1.east) -- ++(0.2,0) node[right] {\scriptsize$\ket{0}$};
\draw[line width = 1] ([yshift=-0.2cm]G1.east) -- ++(0.2,0) node[right] {\scriptsize$\ket{0}$};
\draw[line width = 1] ([yshift=0.2cm]G2.east) -- ++(0.2,0) node[right] {\scriptsize$\ket{0}$};
\draw[line width = 1] ([yshift=-0.2cm]G2.east) -- ++(0.2,0) node[right] {\scriptsize$\ket{0}$};
\end{tikzpicture}
\end{array},
\label{eqn:overlap_diagram}
\end{equation}
where $\bra{\psi_{\text{target}}}\leftrightarrow\begin{array}{c}\begin{tikzpicture}[baseline=-0.55ex]
  \Vertex[x=0,    y=0,size=0.16,color=red!30,style={line width=1pt}]{ta}
  \Vertex[x=0.45, y=0,size=0.16,color=red!30,style={line width=1pt}]{tb}
  \Vertex[x=1.3, y=0,size=0.16,color=red!30,style={line width=1pt}]{tc}
  \foreach \n in {ta,tb,tc}{%
    %\draw[line width=1pt] (\n) -- ++(0.15,0.15);
    \draw[line width=1pt] (\n) -- ++(0,0.2);}
  \Edge[lw=0.55](ta)(tb)
  \Edge[lw=0.55,label=$\cdots$](tb)(tc)
\end{tikzpicture}\end{array}$.

Finally, since our circuit is parametrised in terms of arbitrary two-qubit $\operatorname{SU}(4)$ gates, single-qubit rotations in the circuit structure are free, as they can be absorbed into neighbouring two-qubit gates.
Therefore, the quoted values of infidelity we will use are with respect to the optimal initial single-qubit basis choices, i.e.
\begin{equation}
 \varepsilon_{\star} = 1-\max_{\{v_j\}}
    \bigl|\bra{\psi_{\text{target}}}C(\textstyle\bigotimes_j v_j)\ket{0}^{\otimes N}\bigr|,
  \label{eq:infidelity-frame}
\end{equation}
in which the optimal initial single-qubit basis $\bigotimes_j v_j$ rotations are applied to the initialised register, and absorbed into the first layer of $\operatorname{SU}(4)$ gates. We will henceforth refer to our metric simply as infidelity. The single-qubit gate layer is free only in the two-qubit cost model. On the fault-tolerant axis of Figure~\ref{fig:fig1-efficacy}\textbf{(a)} it is $N$ arbitrary single-qubit rotations and carries a $T$ cost, which is charged there; it is charged identically to both arms, so it shifts both curves by a common additive offset and does not affect their comparison.

\subsubsection{Procrustes optimisation}\label{sec:procrustes}
Once the initial circuit has been constructed (whether by our ASPIRE protocol or the sequential disentangler of Ref.~\cite{ran2020encoding}), there is still substantial room to further optimise gate choices without changing the overall circuit structure. This can be achieved gate-by-gate by constructing and solving a unitary Procrustes problem via tensor-network contraction, an approach first demonstrated in the Evenbly-Vidal algorithm~\cite{evenbly2009algorithms} for the multiscale entanglement renormalisation ansatz (MERA)~\cite{vidal2007entanglement,vidal2008class}, and introduced to the MPD in Ref.~\cite{rudolph2023decomposition}. Unlike MERA optimisation, the objective is linear in the gate being optimised and is therefore exact (optimal for that gate with all others fixed), except in Section~\ref{sec:energy_refinement} where we linearise a quadratic secondary objective with the Evenbly-Vidal procedure. This type of optimisation is a standard approach in tensor-network optimisation that is popular in a variety of contexts~\cite{evenbly2009algorithms,lin2021real,rudolph2023decomposition,ben2024approximate,sims2026multivariate,gibbs2025deep,gibbs2026low,szoldra2026scalable}.

For a circuit $C=U_M\cdots U_2U_1$, in which $U_1$ acts first on $\ket{0}^{\otimes N}$, the environment tensor of the gate in slot $m$ is:
\begin{equation}
     \mathcal{F}_m=\Tr_{\bar U_m}\Bigl[U_{m+1}^\dagger\cdots U_M^\dagger\ket{\psi_{\mathrm{target}}}\bra{0^{\otimes N}}U_1^\dagger\cdots U_{m-1}^\dagger\Bigr],
    \label{eq:environment}
\end{equation}
Diagrammatically, this corresponds to the removal of gate $m$ from the tensor-network diagram for $f$ with its attached legs left open, e.g. for the example in Equation~\ref{eqn:overlap_diagram} the environment for $m=3$ can be written
\begin{equation}
\mathcal{F}_3^\dagger\leftrightarrow
\begin{array}{@{}c@{}}
\begin{tikzpicture}
% target MPS <psi| (left column)
\Vertex[x=0,y=0,size=0.2,color=red!30]{P1}
\Vertex[x=0,y=-0.4,size=0.2,color=red!30]{P2}
\Vertex[x=0,y=-0.8,size=0.2,color=red!30]{P3}
\Vertex[x=0,y=-1.2,size=0.2,color=red!30]{P4}
\Edge[lw=1](P1)(P2)
\Edge[lw=1](P2)(P3)
\Edge[lw=1](P3)(P4)
% brickwork with U_3 cut out (dashed placeholder)
\Vertex[x=0.5,y=-0.2,shape=rectangle,label=$U_4$,position=center,fontsize=\scriptsize,color=blue!15,style={minimum width=0.42cm,minimum height=0.6cm}]{G4}
\Vertex[x=0.5,y=-1.0,shape=rectangle,label=$U_5$,position=center,fontsize=\scriptsize,color=blue!15,style={minimum width=0.42cm,minimum height=0.6cm}]{G5}
\Vertex[x=1.1,y=-0.6,shape=rectangle,color=white,style={minimum width=0.42cm,minimum height=0.6cm,dashed,draw=gray}]{G3}
\Vertex[x=1.7,y=-0.2,shape=rectangle,label=$U_1$,position=center,fontsize=\scriptsize,color=blue!15,style={minimum width=0.42cm,minimum height=0.6cm}]{G1}
\Vertex[x=1.7,y=-1.0,shape=rectangle,label=$U_2$,position=center,fontsize=\scriptsize,color=blue!15,style={minimum width=0.42cm,minimum height=0.6cm}]{G2}
% wires: MPS -> layer nearest the MPS
\draw[line width = 1] (P1) -- ([yshift=0.2cm]G4.west);
\draw[line width = 1] (P2) -- ([yshift=-0.2cm]G4.west);
\draw[line width = 1] (P3) -- ([yshift=0.2cm]G5.west);
\draw[line width = 1] (P4) -- ([yshift=-0.2cm]G5.west);
% wires: between layers (open legs of the cut-out gate)
\draw[line width = 1] ([yshift=0.2cm]G4.east) -- ([yshift=0.2cm]G1.west);
\draw[line width = 1] ([yshift=-0.2cm]G4.east) -- ([yshift=0.2cm]G3.west);
\draw[line width = 1] ([yshift=0.2cm]G5.east) -- ([yshift=-0.2cm]G3.west);
\draw[line width = 1] ([yshift=-0.2cm]G5.east) -- ([yshift=-0.2cm]G2.west);
\draw[line width = 1] ([yshift=0.2cm]G3.east) -- ([yshift=-0.2cm]G1.west);
\draw[line width = 1] ([yshift=-0.2cm]G3.east) -- ([yshift=0.2cm]G2.west);
% initial state |0>
\draw[line width = 1] ([yshift=0.2cm]G1.east) -- ++(0.2,0) node[right] {\scriptsize$\ket{0}$};
\draw[line width = 1] ([yshift=-0.2cm]G1.east) -- ++(0.2,0) node[right] {\scriptsize$\ket{0}$};
\draw[line width = 1] ([yshift=0.2cm]G2.east) -- ++(0.2,0) node[right] {\scriptsize$\ket{0}$};
\draw[line width = 1] ([yshift=-0.2cm]G2.east) -- ++(0.2,0) node[right] {\scriptsize$\ket{0}$};
\end{tikzpicture}
\end{array},
\label{eqn:environment_diagram}
\end{equation}
which resembles Equation~\ref{eqn:overlap_diagram} with $U_3$ `cut out'; contracting $U_m$ back into the diagram recovers $f=\Tr\left[U_m\mathcal{F}_m^\dagger\right]$. Decomposing the environment tensor then yields the optimal gate replacement:

\begin{informalalg} \textbf{ Unitary Procrustes optimisation.} After circuit initialisation, perform $n_s$ sweeps over the whole circuit. Each sweep visits the gates in circuit order, $m=1,\dots,M$, and for each gate: (I) computes $\mathcal{F}_m$ (Equation~\ref{eq:environment}); (II) takes the singular value decomposition $\mathcal{F}_m=W\Sigma V^\dagger$ and sets $U_{\mathrm{new},m}=WV^\dagger$, the unitary that maximises $|\Tr[U^\dagger\mathcal{F}_m]|$~\cite{rudolph2023decomposition}; and (III) replaces $U_m\leftarrow U_{\mathrm{new},m}$.
\label{alg:optimisation}
\end{informalalg}

Where we optimise to convergence, the circuit is first constructed in full and the sweeps of Algorithm~\ref{alg:optimisation} are then repeated until the infidelity $\varepsilon=1-f$ of Equation~\ref{eq:infidelity} stops improving. We measure $\varepsilon$ at a sequence of sweep counts and stop once each of the last three measurement intervals has lowered it by less than a fraction $\delta=10^{-7}$ of its value at the start of that interval, per sweep in the interval. We optimise the whole circuit for a fixed classical budget (Figure \ref{fig:fig1-efficacy}) number of sweeps (Figure \ref{fig:fig2-randomMPS}) to aid comparison between protocols, but note that the most performant protocol in Ref.~\cite{rudolph2023decomposition} re-optimises at the application of each layer.

The environment $\mathcal{F}_m$ is computed as the overlap of two MPS, the ket $U_{m-1}\cdots U_1\ket{0}^{\otimes N}$ and the bra $U_{m+1}^\dagger\cdots U_M^\dagger\ket{\psi_{\mathrm{target}}}$, contracted over every site except the two on which $U_m$ acts; each gate application is followed by an SVD truncation of the bond dimension to $\chi_w$, twice the bond dimension of the target, so both sides stay in MPS form.

Contracting the two MPS from both ends towards the sites of $U_m$ gives the $4\times4$ matrix $\mathcal{F}_m$ at cost $\mathcal{O}(N\chi_w^3)$. A sweep visits the gates in circuit order, so the ket advances by one gate application per step, using the gate just updated, while the bras for every $m$ are computed once at the start of the sweep and cached. A sweep over $M$ gates therefore costs $\mathcal{O}(M)$ gate applications and $\mathcal{O}(M)$ environment contractions, at the price of storing $M$ MPS; a gate of range $|i-j|$ is applied through a SWAP chain and costs $\mathcal{O}(|i-j|\,\chi_w^3)$.

\subsubsection{Warm-started Procrustes optimisation}
Beyond the removable correlation $R_{ij}$ of Section~\ref{sec:protocol}, the Procrustes
step itself supplies a gate-selection score. Let $C$ be the current
preparation circuit and $\ket{w}=C^\dagger\ket{\psi_{\mathrm{target}}}$ the
residual state that ASPIRE disentangles. Consider a new gate acting on $(i,j)$ through $U^\dagger_{ij}$, so that the circuit becomes $CU_{ij}$, and its overlap with the
target is
$\bra{0^{\otimes N}}U_{ij}^\dagger\ket{w}=\Tr\bigl[U_{ij}^\dagger\mathcal{F}_{ij}\bigr]$,
with environment
$\mathcal{F}_{ij}=\Tr_{\overline{ij}}\bigl[\ket{w}\bra{0^{\otimes N}}\bigr]$.
With this, define
\begin{equation}
  P(i,j)\ =\ \max_{U_{ij}}\bigl|\bra{0^{\otimes N}} U_{ij}^{\dagger}
  C^{\dagger}\ket{\psi_{\mathrm{target}}}\bigr|
  \ =\ \lVert\mathcal{F}_{ij}\rVert_{1},
  \label{eq:procrustes-score}
\end{equation}
the largest overlap any two-qubit gate on that support can reach.
This is step~(II) of Algorithm~\ref{alg:optimisation} in closed form, where the maximum of
$|\mathbf{Tr}[U^{\dagger}\mathcal{F}_{ij}]|$ over $\mathrm{SU}(4)$ is the sum of the singular values of $\mathcal{F}_{ij}$, and since $\bra{0^{\otimes N}}$ is a product state $\mathcal{F}_{ij}$ has rank one, so that sum is a single singular value and no sweep is required. $P$ therefore ranks candidate pairs by the fidelity attainable on that support, and is exact.

Since $P$ is an overlap with $\ket{0\cdots0}$, it ranks candidates only once the residual has weight there, and ties identically when symmetry forbids it. In the $S_z=0$ Heisenberg sector every configuration within one flipped pair of all-up is excluded, so all $N(N-1)/2$ candidates vanish and the selector emits no gate. To avoid this one can adopt a selection schedule where calculating $P$ begins only after a specified number of ASPIRE warm-start layers have been applied. This then creates an optimisation protocol where after the warm-start layers the optimisation effort is built directly into the gate-selection for the final sequence of gates. 

We tested this on our $4\times 3$ TFIM $h=1$ and Heisenberg ground state, and random $N=12$ MPS systems with $\chi=64$ and $33$ two-qubit gates. Sweeping the number $k$ of warm-start layers, the best-$k$ arm improves on Algorithm~\ref{alg:optimisation} at matched gate count and matched contraction count on two of the three targets, by $1.52\times$ (TFIM) and $1.97\times$ (Heisenberg) reduction in infidelity, and ties on the random MPS at $0.98\times$. For the Heisenberg system we found optimal performance at $5$ layers, $2$ layers for the random MPS, and $0$ for the TFIM system. The TFIM residual already has $f=0.199$ at the first layer, so the residual already overlaps with the zero state and no warm-up is needed $(k=0)$. The gate itself is a single Procrustes coordinate step and not a new construction, and here we build it directly into the selection rule with free support.

\subsubsection{Energy-refinement optimisation of ground states}\label{sec:energy_refinement}
Thus far, Section~\ref{sec:optimisation} has been concerned with the preparation of a given target state $\ket{\psi_{\text{target}}}$, which has a tractable MPS representation. For end-to-end preparation of Hamiltonian ground states, the true ground state $\ket{\psi_{\text{exact}}}$ is typically unknown, and may require an exponentially-large bond dimension to represent exactly.
Therefore, one typically starts by approximating the ground state $\ket{\psi_{\text{exact}}}$ by an MPS $\ket{\psi_{\text{MPS}}}$ with DMRG~\cite{white1992density,white1993density,schollwock2011density}. 
Even in cases where $\ket{\psi_{\text{exact}}}$ requires an exponentially-large bond dimension, a tractably sized MPS can yield a high-probability guide state for QPE~\cite{berry2025rapid} with good but imperfect overlap; improving convergence or increasing the bond dimension of the MPS can enhance this overlap.
Once the MPS approximation has been computed, in MPD methods one then attempts to approximately prepare $\ket{\psi_{\text{exact}}}$ by preparing a circuit $C$ that maximises the fidelity
\begin{equation}
f_{\text{MPS}}\equiv\bigl|\bra{\psi_{\text{MPS}}}C\ket{0}^{\otimes N}\bigr|,
\end{equation}
whether by increasing layer count~\cite{ran2020encoding} or by the unitary Procrustes optimisation outlined in Section~\ref{sec:procrustes} \cite{rudolph2023decomposition}. As noted therein, this fidelity, and its associated environment tensors~\eqref{eq:environment}, are efficiently computable by MPS-MPS contraction. The better the DMRG approximation $\ket{\psi_{\text{MPS}}}$ is to the true ground state $\ket{\psi_{\text{exact}}}$, the more successful this strategy will be, however, the DMRG overlap $|\bra{\psi_{\text{exact}}}\ket{\psi_{\text{MPS}}}|$ bottlenecks the end-to-end success of state preparation; one ideally would wish to optimise with respect to the circuit's fidelity to the exact ground state
\begin{equation}
f_{\text{exact}}\equiv\bigl|\bra{\psi_{\text{exact}}}C\ket{0}^{\otimes N}\bigr|,
\end{equation}
which cannot be efficiently computed if the true ground state $\ket{\psi_{\text{exact}}}$ has an exponentially-large bond dimension, as is
expected for systems where QPE could offer an advantage. Nonetheless, we can attempt to improve $f_{\text{exact}}$ indirectly by optimising with respect to quantities that are efficiently computable. For ground-state preparation workflows, once Algorithm~\ref{alg:optimisation} has been run to convergence, further improvements in $f_{\text{exact}}$ can be sought by switching the optimisation objective from maximising the overlap with the approximate MPS to minimising the energy with respect to the true Hamiltonian:
\begin{equation}
  E[C]\equiv\bra{0}^{\otimes N}C^{\dagger}HC\ket{0}^{\otimes N}
  \label{eq:energy-objective}
\end{equation}
which, since an efficient MPO approximation
$H\leftrightarrow\begin{array}{c}\begin{tikzpicture}[baseline=-0.55ex]
  \Vertex[x=0,    y=0,size=0.16,color=violet!30,style={line width=1pt}]{ta}
  \Vertex[x=0.45, y=0,size=0.16,color=violet!30,style={line width=1pt}]{tb}
  \Vertex[x=1.3, y=0,size=0.16,color=violet!30,style={line width=1pt}]{tc}
  \foreach \n in {ta,tb,tc}{%
    \draw[line width=1pt] (\n) -- ++(0.15,0.15);
    \draw[line width=1pt] (\n) -- ++(-0.15,-0.15);}
  \Edge[lw=0.55](ta)(tb)
  \Edge[lw=0.55,label=$\cdots$](tb)(tc)
\end{tikzpicture}\end{array}$ is required to perform DMRG in the first place, can be computed efficiently as
\begin{equation}
E[C] =
\begin{array}{c}
\begin{tikzpicture}
% ket (top) row: MPS of C|0>^{\otimes N}; sites 1, 2, N
\Vertex[x=0,y=0.3,size=0.2,color=blue!30]{K1}
\Vertex[x=0.5,y=0.3,size=0.2,color=blue!30]{K2}
\Vertex[x=1.4,y=0.3,size=0.2,color=blue!30]{KN}
% MPO (middle) row: H
\Vertex[x=0,y=0,size=0.2,color=violet!35]{W1}
\Vertex[x=0.5,y=0,size=0.2,color=violet!35]{W2}
\Vertex[x=1.4,y=0,size=0.2,color=violet!35]{WN}
% bra (bottom) row: conjugate MPS
\Vertex[x=0,y=-0.3,size=0.2,color=red!30]{B1}
\Vertex[x=0.5,y=-0.3,size=0.2,color=red!30]{B2}
\Vertex[x=1.4,y=-0.3,size=0.2,color=red!30]{BN}
% physical legs: ket -> MPO input, MPO output -> bra
\Edge[lw=1](K1)(W1)
\Edge[lw=1](K2)(W2)
\Edge[lw=1](KN)(WN)
\Edge[lw=1](W1)(B1)
\Edge[lw=1](W2)(B2)
\Edge[lw=1](WN)(BN)
% ket-row bonds (elided gaps use \cdots)
\Edge[lw=1](K1)(K2)
\Edge[lw=1,label=$\cdots$](K2)(KN)
% MPO-row bonds
\Edge[lw=1](W1)(W2)
\Edge[lw=1,label=$\cdots$](W2)(WN)
% bra-row bonds
\Edge[lw=1](B1)(B2)
\Edge[lw=1,label=$\cdots$](B2)(BN)
\end{tikzpicture}
\end{array},
\label{eqn:energy_diagram_form}
\end{equation}
via standard MPO-MPS contraction methods~\cite{schollwock2011density}, where $C\ket{0}^{\otimes N}\leftrightarrow\begin{array}{c}\begin{tikzpicture}[baseline=-0.55ex]
  \Vertex[x=0,    y=0,size=0.16,color=blue!30,style={line width=1pt}]{ta}
  \Vertex[x=0.45, y=0,size=0.16,color=blue!30,style={line width=1pt}]{tb}
  \Vertex[x=1.3, y=0,size=0.16,color=blue!30,style={line width=1pt}]{tc}
  \foreach \n in {ta,tb,tc}{%
    %\draw[line width=1pt] (\n) -- ++(0.15,0.15);
    \draw[line width=1pt] (\n) -- ++(0,-0.2);}
  \Edge[lw=0.55](ta)(tb)
  \Edge[lw=0.55,label=$\cdots$](tb)(tc)
\end{tikzpicture}\end{array}$ and $\bra{0}^{\otimes N}C^\dagger\leftrightarrow\begin{array}{c}\begin{tikzpicture}[baseline=-0.55ex]
  \Vertex[x=0,    y=0,size=0.16,color=red!30,style={line width=1pt}]{ta}
  \Vertex[x=0.45, y=0,size=0.16,color=red!30,style={line width=1pt}]{tb}
  \Vertex[x=1.3, y=0,size=0.16,color=red!30,style={line width=1pt}]{tc}
  \foreach \n in {ta,tb,tc}{%
    %\draw[line width=1pt] (\n) -- ++(0.15,0.15);
    \draw[line width=1pt] (\n) -- ++(0,0.2);}
  \Edge[lw=0.55](ta)(tb)
  \Edge[lw=0.55,label=$\cdots$](tb)(tc)
\end{tikzpicture}\end{array}$. The environment tensors for each gate in the circuit can be computed similarly. A small complication arises in the fact that while the objective in the fidelity-based optimisation of Section~\ref{sec:procrustes} is linear in the gate tensors, allowing it to be posed as a standard unitary Procrustes problem, this objective is quadratic in the gate tensors. We overcome this by borrowing the linearisation procedure from the Evenbly-Vidal algorithm for MERA optimisation~\cite{evenbly2009algorithms}, temporarily regarding the bra and ket copies of the gate as independent and alternating between optimising each copy while keeping its dual fixed. This protocol requires only additional classical optimisation, and does not alter circuit structure or depth.

\begin{figure}
    \centering
    \includegraphics[width=\linewidth]{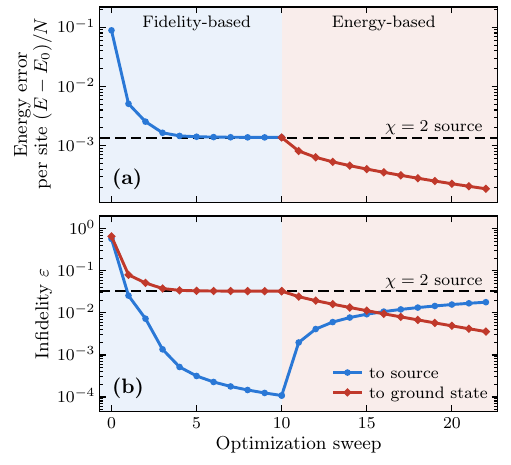}
    \caption{\textbf{Toy model of the energy refinement.} A $N=8$ critical TFIM ground state is prepared with $21$ two-qubit gates via the unitary disentangler protocol from a deliberately crude $\chi=2$ single-sweep DMRG state. Comparing against $E_0$ from exact diagonalisation \textbf{(a)} Energy error per site, relative to $E_0$ from exact diagonalisation, against optimisation sweep; at sweep 10 the protocol switches from Algorithm~\ref{alg:optimisation} to energy refinement. \textbf{(b)} shows the infidelity with respect to the true ground state (red) and the DMRG state (blue), with the same optimisation switch halfway.}
    \label{fig:toy-energy}
\end{figure}

Although minimising $E[C]$ is only a heuristic to improve $f_{\text{exact}}$, our numerical tests have shown that it succeeded in all cases we
tested . We note that this can only be confirmed at systems small enough to allow for exact state-vector-based computation of $f_{\text{exact}}$; for larger systems this could be confirmed post-hoc via e.g. verification of the QPE probability distribution. We present an illustrative toy example of the method in Figure~\ref{fig:toy-energy}: a small MPS approximating a TFIM ground state is first optimised with the fidelity-based methods of Section~\ref{sec:procrustes}, and then refined with the energy-based methods of this subsection. We observe that upon the switch to energy-based methods, the energy $E[C]$ decreases, the MPS fidelity $f_{\text{MPS}}$ \emph{worsens}, but the exact fidelity $f_{\text{exact}}$ improves as a consequence, showing that energy-based refinement can lead to better end-to-end state preparation even when the optimisation objectives typically used in the literature worsen. We note for completeness that
energy refinement of circuit tensor networks is established ~\cite{haghshenas2022variational, lin2021real} and is state-preparation protocol-agnostic; Figure~\ref{fig:toy-energy} uses MPD circuits, and nothing in the procedure is specific to either selection rule. We include it as post-processing available to both arms rather than as a component of ASPIRE.

\section{Results} \label{sec:results}

\subsection{Two-dimensional lattice Hamiltonians}
\label{subsec:lattice}
\emph{Conventions.} Every transverse-field Ising target in this work is the
ground state of
\begin{equation}
  H=-J\sum_{\langle ab\rangle}\frac{\sigma^{z}_{a}}{2}\frac{\sigma^{z}_{b}}{2}
   -h\sum_{a}\frac{\sigma^{x}_{a}}{2},
  \qquad J=1,
  \label{eq:tfim}
\end{equation}
on the stated lattice with open boundaries, so $g=2h$ in the Pauli convention: the
$h=2$ used below sits above the square-lattice critical point
$h_c=g_c/2\approx1.52$ ($g_c\approx3.044$~\cite{blote2002cluster}), in the
paramagnetic phase, while $h=0.5$ is the critical chain of
Figure~\ref{fig:toy-energy}. The $4\times3$ warm-start test of the preceding
section is the one target taken at $h=1$, that is $g=2$, on the ordered side of
$g_c$; these three fields are the only ones this work reports. Heisenberg targets are the
isotropic antiferromagnet:

\begin{equation}
   H=J\sum_{\langle ab\rangle}\frac{\vec\sigma_a}{2}\cdot\frac{\vec\sigma_b}{2}, \quad J=1.
   \label{eq:heisen}
\end{equation}
Ground states are obtained by DMRG on the snake-ordered chain. 
Their construction is summarised in Appendix~\ref{app:targets}.

The clearest demonstration of the advantage of
long-ranged gates is provided by Hamiltonians with long-ranged correlations that
the protocol can exploit \cite{chan2011density,baiardi2020density,menczer2024tensor,menczer2024parallel,xiang2024distributed,legeza2025orbital,white2007neel,yan2011spin,stoudenmire2012studying,zheng2017stripe,zhou2020limits}.Mapping two-dimensional lattices to an MPS is a clear example of this, where the snake mapping picks out which lattice sites become nearest-neighbours. Indeed the majority of the direct connections present in a 2D lattice will be lost when certain snake-mappings are applied. The MPD is limited to targeting these nearest-neighbour connections in the MPS, whereas ASPIRE selects gates for all the possible pairs of sites, thereby directly targeting the long-range interactions.

Figure~\ref{fig:fig1-efficacy} \textbf{(a)} compares the two protocols preparing the $10\times10$ 2D TFIM ground state, on a common fault-tolerant axis measuring T-gate count up to $7 \times 10^5$. The MPD's infidelity struggles to push past $1.24\times10^{-1}$ with its unoptimised state-preparation circuit, whereas ASPIRE continues to descend to $1.61\times10^{-2}$, an improvement factor of $7.7$ at equal $T$-count. The mechanism is geometric: under an $L_x$-wide snake the lattice's vertical bonds acquire ranges $1,3,\dots,2L_x-1$, almost every one of which is outside the staircase's reach and inside ASPIRE's. The solid lines add the optimisation of Algorithm~\ref{alg:optimisation} at matched classical cost. Each data-point is refined with $5\times10^{12}$ operations, determined from charging $\chi^3$ for each two-site operation accounting for the different bond-dimensions used at different sites and for different protocols. Since a gate range of $\ell$ costs $2\ell-1$ operations, ASPIRE affords fewer sweeps than MPD, and both protocols afford fewer sweeps the deeper the initial circuit was. ASPIRE uses $15$ sweeps at $156$ two-qubit gates, and one at $999$ gates, stopping there. MPD improves by a factor of $4.1$ at $1548$ gates, and yet unoptimised ASPIRE achieves $1.7$ to $2.0$ times lower infidelity at the deepest points. Optimised MPD is not monotone due to the decreasing number of sweeps going $15,\ 12, \ 10$ for its last three data-points, which we highlight by increasing the transparency of the last point. This is a result of the constant classical budget, as at a fixed number of sweeps, the deeper circuits reach lower infidelities and vice-versa.

In Figure~\ref{fig:fig1-efficacy}\textbf{(b)} the advantage factor defined as the infidelity ratio of unoptimised MPD and ASPIRE infidelities is tracked for TFIM lattices with sizes $5\times5-10\times10$. The ratio is calculated at the point of attaining $1.61 \times 10^{-2}$ infidelity which is the final infidelity of ASPIRE for the $N=100$ lattice. ASPIRE
reaches this accuracy at gate counts equal to $2,4,7,12,24$ MPD layers for the
respective sizes, and the MPD is afforded the same budget, and so the gate count is not constant but rises with system size as expected. We observe that the advantage rises linearly in the MPS length, but note that this was sensitive to the number of gates used for each datapoint and was calculated from the unoptimised circuits.

\begin{figure}
    \centering
    \includegraphics[width=\linewidth]{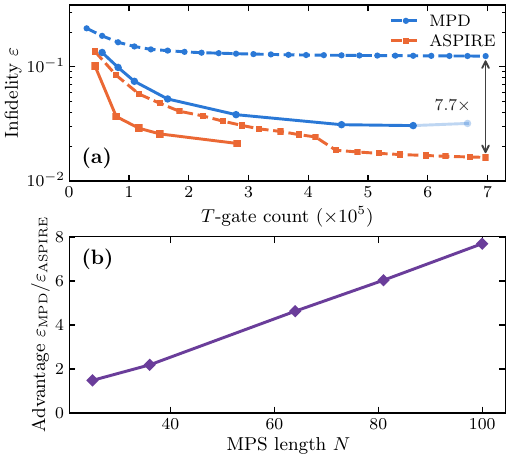}
    \caption{\textbf{2D TFIM lattice ground-state preparation with MPD and ASPIRE.} \textbf{(a)} The MPD (blue) versus ASPIRE (orange) preparing the $10\times10$ 2D TFIM ground state ($N=100$ sites, $h=2$), both gate-matched at $2376$ two-qubit gates and cast onto a fault-tolerant $T$-gate count with the \texttt{TRASYN} synthesis model~\cite{hao2026reducing}, the specialised $U_3$ synthesis of Ref.~\cite{gibbs2026low}, total synthesis error $10^{-3}$. Dashed lines are the unoptimised constructions. Solid lines refine the circuit at each marked gate count with Algorithm~\ref{alg:optimisation}, each point separately and with the same budget of $5\times10^{12}$ bond-resolved $\chi^3$ operations for both protocols. ASPIRE's solid line ends at $999$ two-qubit gates, beyond which the budget buys less than one sweep of its circuit. The MPD's solid line is faded where its optimised infidelity rises with depth, which reflects fewer affordable sweeps rather than the circuit getting worse. \textbf{(b)} the same comparison at fixed ASPIRE infidelity of $1.61 \times 10^{-2}$, the highest quality every system attains. ASPIRE requires a gate-count of  $2$, $4$, $7$, $12$ and $24$ MPD-layers at $N = 25$, $36$, $49$, $64$, $81$ and $100$ respectively. The advantage factor is the ratio of the MPD's infidelity to ASPIRE's, both read at the two-qubit gate count at which ASPIRE first reaches $1.61\times10^{-2}$, with the MPD's budget rounded up to a whole number of layers.}
    \label{fig:fig1-efficacy}
\end{figure}

\subsection{Optimisation comparison}

The inclusion of optimisation by way of Algorithm~\ref{alg:optimisation} introduces another axis of comparison. This comparison must account both for the efficacy and classical computational cost of the optimisation when applied to the two protocols. In a matrix-product implementation a gate of range $\ell$ is realised as $2\ell-1$ elementary two-site operations, and a refinement sweep visits each gate with both a forward and a backward environment, so a sweep over a circuit of $M$ gates costs $2\sum_{g=1}^{M}\bigl(2\ell_g-1\bigr)=4\sum_g|i_g-j_g|-2M$ such operations, where $\ell_g=|i_g-j_g|$ is the range of gate $g$.

In Figure~\ref{fig:fig2-randomMPS} we compare the effect of optimisation to protocol convergence \textbf{(a-c)} for a 2D Heisenberg ground state, a 2D TFIM ground state, and a random MPS as control. The experiment is deliberately under-parameterised (as it would  otherwise perfectly prepare the state), and each circuit is constructed in full and then we attempt to optimise to convergence, as defined after Algorithm~\ref{alg:optimisation}. We pick rectangular $2\times \frac{N}{2}$ lattices snake mapped in two straight lines to showcase how the resultant long-ranged entanglement is better exploited by the ASPIRE protocol also when optimisation is applied. As a control we also test on random MPS of the same size, in which no such clear structure is present. For these examples we chose an MPS ordering adversarial to MPD by snake-mapping along the longer side of the 2D rectangular lattices. This places precisely the long-range correlations that our algorithm targets, and does not claim to be the ordering a real implementation of MPD would use on these systems. Instead this demonstrates the general point that in regimes with long ranged correlations and a limited number of MPD layers, ASPIRE can attain lower infidelities both from its circuit construction and from the richness of the optimisation space it affords to Algorithm~\ref{alg:optimisation}. For a comparison with the MPS ordering as favourable as possible to the MPD, we refer the reader to the square lattice of Figure~\ref{fig:fig1-efficacy}. Figure~\ref{fig:fig2-randomMPS} \textbf{(a-b)} shows that optimised ASPIRE converges to a lower infidelity floor when compared to a gate-matched optimised MPD with $2$ to $5$ layers for the $7\times2$ Heisenberg and TFIM ground states, by a factor that grows with depth from $2.8$ to $9.8$ on the Heisenberg lattice, and is $54$ at two layers and at least $500$ beyond on the TFIM (where the ASPIRE datapoints hit the sweep cap instead of converging). In the random MPS control of \textbf{(c)} the MPD leads at two to four layers and the two protocols tie at five, likely due to the lack of long-ranged structure for ASPIRE to exploit. This supports the geometric interpretation as the MPS carries no lattice scale for increased connectivity radius to exploit. From three layers up, optimisation widens the gap between MPD and ASPIRE on both structured lattices.

\begin{figure}
    \centering
    \includegraphics[width=\columnwidth]{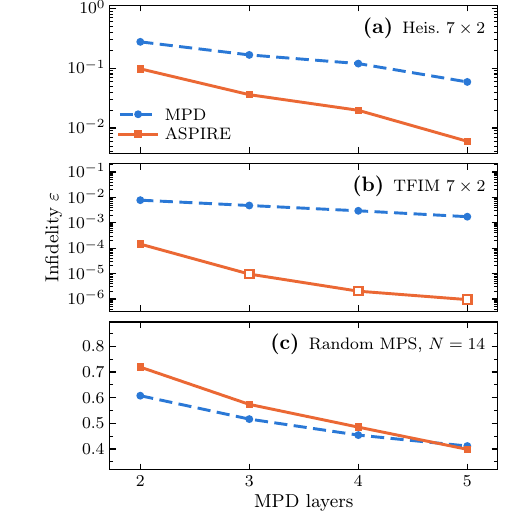}
    \caption{\textbf{Optimisation effectiveness of MPD and ASPIRE.} The infidelity of the protocols optimised to convergence at a gate budget of $2$ to $5$ MPD layers is compared for the $7\times2$ Heisenberg \textbf{(a)} and TFIM \textbf{(b)} ground states, and the $N=14$ random MPS with $\chi=64$ \textbf{(c)}, the last a geometric mean over three seeds. Hollow markers stopped at the optimiser's cap of $3\times10^{4}$ sweeps and bound the converged infidelity from above.}
    \label{fig:fig2-randomMPS}
\end{figure}

\subsection{Connectivity}
\label{subsec:connectivity}

To demonstrate the effect of increased connectivity in the hardware, one can restrict ASPIRE to a varying connectivity radius $r\geq|i-j| \ \forall (i,j)$ MPS sites that ASPIRE connects with a gate. As demonstrated, unrestricted ASPIRE can outperform the MPD for certain systems, and varying $r$ shows how this advantage changes with increased connectivity. We fix one definition of the onset radius and use it in every panel: the smallest radius at which the median over seeds of ASPIRE's per-seed advantage over a gate-matched MPD is positive, at the budget stated for that panel. The median rather than the mean, because a small and reproducible fraction of seeds stalls in a greedy local optimum and holds the mean negative at any seed count. ASPIRE's performance is observed to increase with radius, but the monotonicity is not guaranteed as a larger move set does not guarantee a better greedy trajectory at fixed budget, and Appendix~\ref{app:tau} records a related case in which the apparently better-motivated ranking statistic loses. The onset indicates how much connectivity the hardware platform must be able to realise to naively utilise ASPIRE's advantage. In practice, ASPIRE can still be advantageous without the full connectivity required, if the benefit of out-of-range connections outweighs the cost of the SWAP gates to bring them within hardware range. Furthermore, one can track the onset radius as a function of system size, where sublinear growth of the onset radius with N would show that ASPIRE outperforms the MPD with less than all-to-all connectivity. Lastly, tracking the onset radius as a function of bond dimension of the initial MPS allows one to investigate if the longer-ranged interactions of a higher-$\chi$ MPS correspond to a larger onset radius.

In Figure~\ref{fig:fig3-connectivity-spectra} such a connectivity analysis is performed for the 2D Heisenberg and  TFIM systems. In \textbf{(a)} the $N=12$ spin ground states are prepared by the MPD and by ASPIRE at a range of connectivity
radii; the onset radius is $6$ for the Heisenberg target and $r=5$ for TFIM \textbf{(b)},  by tracking the infidelity as a function of two-qubit gates. Performance is seen to improve with radius throughout, though not strictly monotonically due to the noisy performance of the greedy descent. This demonstrates that about half of the all-to-all radius was sufficient for ASPIRE to achieve a lower infidelity at fixed gate-count than MPD.

Continuing the investigation, Figure~\ref{fig:fig3-connectivity-trends} tracks the onset radius, defined as the smallest radius at which ASPIRE beats MPD, for the 2D Heisenberg at three different lattice lengths and the random MPS. In \textbf{(a)} the system size for all systems is varied and the onset radius is measured as the median over $8$--$16$ seeds, with the $4\times5$ Heisenberg system requiring a radius of $8$ and the $N=16$ random MPS $r=7$, in every case well below the all-to-all bound $r=N-1$. The $3\times5$ TFIM ($r=4$) and $N=16$ random MPS ($r=7$) controls follow the same pattern. Also shown, as a dashed rule per strip width, is $2L_x-1$: the longest range a gate set needs in order to reach every lattice bond once the lattice is snake-mapped to the MPS representation.

This introduces three points. Firstly, under an $L_x$-wide snake the lattice's vertical bonds have ranges $1,3,\dots,2L_x-1$, so the radius at which reach stops buying anything should be set by the strip width and not by $N$ or by $\chi$. The measured onsets are close to consistent with that: the $L_x=3$ Heisenberg series of Figure~\ref{fig:fig3-connectivity-trends}\textbf{(a)} runs $4,4,5,5$ over $N=9,\dots,18$ against $2L_x-1=5$, the $L_x=4$ series runs $5,5,6,8$ against $2L_x-1=7$, and the $L_x=5$ series runs $8,8$ at $N=15,20$ against $2L_x-1=9$. Every onset sits within two radii of its ceiling, and the only point above it is the $L_x=4$ lattice at $N=20$, at $8$ against $7$. The lack of perfect adherence can likely be attributed to the inconsistent performance of the greedy algorithm, and the imperfections of the entanglement metric used discussed in Section~\ref{sec:conclusion}. The $\chi$ sweep of \textbf{(b)}, which holds the width fixed at $L_x=4$, is flat at $7$--$9$. The random matrix product state, which has no such scale, drifts with $N$ ($3,5,7$, not shown in Figure~\ref{fig:fig3-connectivity-trends}\textbf{(a)}) and with $\chi$ in \textbf{(b)}.

Secondly, restricted to nearest neighbours, ASPIRE performs markedly worse than
the MPD, with $4.5\times10^{-1}$ against $1.7\times10^{-1}$ on the Heisenberg target and $9.5\times10^{-3}$ against $2.0\times10^{-3}$ on TFIM.  The two protocols differ in reach and also in objective, and at $r=1$ only the second difference is active. An explanation for this is that greedy descent on single-site entropies stalls on pairs whose residual correlation is classical in the optimal product basis (Corollary~\ref{cor:sat} Appendix~\ref{app:gate}), whereas the disentangler's truncation objective still sees bond entanglement. Increasing the reach gives ASPIRE pairs
on which the removable correlation is not saturated, which is why the ordering
reverses once $r$ passes the onset

Lastly, the different panels use different notions regarding seeds. For the random matrix product states a seed draws a genuinely
different target. For the lattice ground states it does not: at $N=12$ and $\chi=64$
the DMRG state is exact, so the seeds vary the
tie-break among exactly degenerate candidate scores. The $4\times3$ Heisenberg
lattice's symmetry makes the top four candidate pairs exactly tied, and the
two-site reduced state carries an $SU(2)$ triplet degeneracy, so the spread on
those series is tie-break variance and not a physical error bar. These effects are not negligible and likely explain why the onsets in Figure~\ref{fig:fig3-connectivity-trends} do not vary smoothly from one sweep point to the next.

\begin{figure}
    \centering
    \includegraphics[width=1.0\columnwidth]{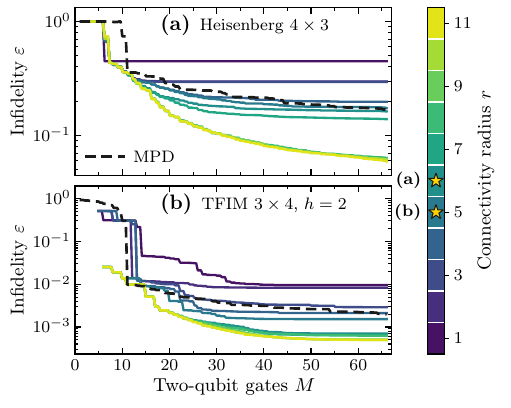}
    \caption{\textbf{Increasing ASPIRE's connectivity radius.} The infidelity over the two-qubit gate count for the $4\times3$ Heisenberg \textbf{(a)} and $3\times4$ TFIM ($h=2$) \textbf{(b)} ground states snake-mapped to MPS, with ASPIRE at connectivity radius $r=1,\dots,11$ (colour coded) against the MPD staircase (black, dashed) at the same $66$-gate budget, averaged over $16$ seeds. The star on the colour bar marks the radius at which ASPIRE first beats the MPD on final infidelity: $r=6$ for the Heisenberg target, $r=5$ for TFIM. A shared budget of $66$ gates was applied to all protocols without optimisation.}
    \label{fig:fig3-connectivity-spectra}
\end{figure}

\begin{figure}
    \centering
    \includegraphics[width=1.0\columnwidth]{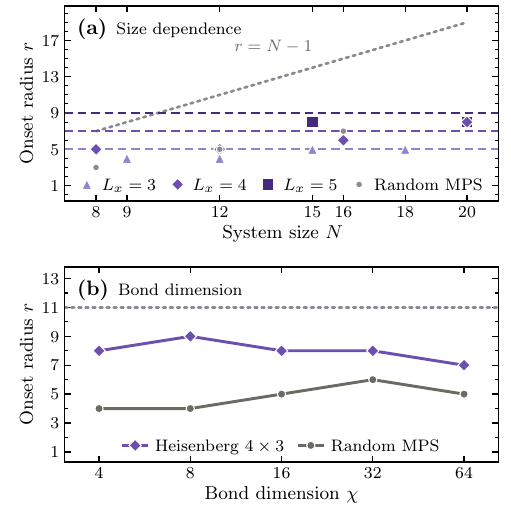}
    \caption{\textbf{First ASPIRE radius to beat MPD as a function of $N$ or $\chi$.} \textbf{(a)} the first radius at which ASPIRE reaches per-seed median advantage over MPD is shown against the system size, for three fixed-width Heisenberg strips ($L_x=5,4,3$, $8$--$16$ seeds per series); the all-to-all limit $r=N-1$ is dotted. The dashed rule in each series' colour marks that strip's $2L_x-1$, the reach needed to cover every bond of an $L_x$-wide snake; it is a prediction, not a fit, and the $L_x=4$ strip at $N=20$ is the one onset above its own rule. The TFIM ($L_x=3$) and random-MPS ($\chi=8$) controls at $N=18$/$16$ follow the same trend. In \textbf{(b)} the onset radius is shown against the target bond dimension at $N=12$, for the Heisenberg snake and for random matrix-product states ($32$ seeds each). All runs were afforded a budget of $6$ \textbf{(a)} and $8$ \textbf{(b)} MPD layers for their gate-count without optimisation.}
    \label{fig:fig3-connectivity-trends}
\end{figure}

\subsection{Routing and gate noise}
\label{subsec:noise}

In the case when the system-to-MPS mapping contains MPS bonds between sites that are not connected on the hardware, the MPD can only be realised by paying for a SWAP network. ASPIRE can instead encode the hardware connectivity into its connectivity matrix, such that it only selects two-qubit gates that are realisable in the hardware. One can also run ASPIRE at full-range and use SWAP networks to realise any generated gates that are beyond the hardware's capacity. In Figure~\ref{fig:decomposition} of Appendix~\ref{app:connectivity} this was done on $N=12$ systems of the Heisenberg, TFIM, and random MPS as before, simulating the most adversarial case of hardware only capable of realising nearest-neighbour connections. The result was a tie with MPD for the spin systems and a slight loss in the random case, all redeemed with more hardware connectivity.

We take the heavy-hex lattice found in some leading superconducting quantum processors as an example of the MPD to hardware gates mismatch problem \cite{Chamberland2020topological}, which places physical qubits on both the vertices and edges of a hexagonal tiling, reducing the maximum vertex degree to at most 3, see Figure~\ref{fig:overview} \textbf{(d)}. This topology has two key features: firstly it is an example of a hardware architecture that realises slightly more than nearest-neighbour connectivity in hexagon corners, and secondly it does not permit a snake-mapping of all sites without "jumping" outside of the lattice to continue when realised on more than 2 hexagons. The first feature makes it ideal for exploring the efficacy of ASPIRE in the connectivity regime that is richer than nearest-neighbour but still far from all-to-all, and the second feature means that the MPD must pay a penalty of SWAP gates whenever an MPS bond is not realised natively in the hardware. 

The current largest heavy-hex device is that of Ref.~\cite{patra2024efficient}, with 1121 qubits. A $i\times j$ heavy-hex graph has $n=5ij+4(i+j)-1$ qubits and longest embeddable line $l_{\max}=4(ij+i+j)+1$, i.e.\ $4n/5$ to leading order \cite{weidenfeller2022scaling}. The lattice is bipartite between its $V=2(ij+i+j)$ vertex qubits and its $n-V$ edge qubits. An ordering broken by $J$ jumps is $J+1$ simple paths, and a path
alternating the two classes contains at most one more edge qubit than corner
qubits, so $n-V\le V+J+1$ and hence
\begin{equation}
  J\ \ge\ n-2V-1\ =\ ij-2\ =\ n-l_{\max}.
  \label{eq:jumps}
\end{equation}
Each jump costs at least two SWAPs, and an MPD layer of $n-1$ emitted gates
plus $2J$ routing SWAPs carries a routing overhead of
\begin{equation}
  1+\frac{2J}{\,n-1\,}\ \xrightarrow[\ ij\to\infty\ ]{}\ 1.4 ,
  \label{eq:routing-overhead}
\end{equation}
using $J/(n-1)\to1/5$. The $N=30$ qubit patch would not realise this asymptotic limit by itself, but we adopt the asymptotic values as the present numerics and analysis are aimed at large-scale frontier quantum processors with up to $\mathcal{O}(10^3)$ physical qubits.

 When accounting for gate-noise, the advantage of our protocol at fixed gate-count generally increases. In the case of the MPS mapping costing the MPD additional SWAP gates this is increased further. ASPIRE has the ability to be selective in which gates it applies, applying fewer gates when the correlation they remove does not outweigh the noise-penalty of applying the gate. With each native two-qubit gate applying a depolarising channel of rate $p$, a natural admission rule accepts a gate on $(i,j)$ when the decrease it produces in the certificate of Theorem~\ref{thm:cert} exceeds the error it injects; since the certificate is an upper bound, this rule is a heuristic rather than a guarantee. This predicts an infidelity as a function of two-qubit gates that initially tracks the noiseless case but diverging to higher infidelities, with the application of each gate, and to a greater extent the stronger the noise. The application of two-qubit gates will lower the infidelity, albeit less than the noiseless case, until the infidelity reduction of a given gate is exceeded by the noise inherent in that gate, at which point an infidelity minimum will be reached and the infidelity will begin to trend upwards.

The kicked-Ising targets are the states of the kicked transverse-field Ising evolution studied on IBM's heavy-hex processors~\cite{kim2023evidence, tindall2024efficient}, run for two steps on the couplers of the device itself,
\begin{equation}
  \begin{aligned}
  \ket{\psi(\theta_h)}&=\bigl[U_{ZZ}\,U_X(\theta_h)\bigr]^{2}\ket{0}^{\otimes N},\\
  U_X(\theta_h)&=\prod_{a}e^{-i\theta_hX_a/2},\quad
  U_{ZZ}=\prod_{\langle ab\rangle}e^{i\pi Z_aZ_b/4},
  \end{aligned}
  \label{eq:kicked-ising}
\end{equation}
where $\langle ab\rangle$ runs over the couplers of the heavy-hex fragment. The kick angle $\theta_h$ interpolates between a product state at $\theta_h=0$ and a stabiliser state at $\theta_h=\pi/2$, where every gate of Equation~\ref{eq:kicked-ising} is Clifford. We use $\theta_h=0.4$ on the $N=30$ fragment and $\theta_h=0.6$ on the $12$-qubit cell of Appendix~\ref{app:noise}.

In Figure~\ref{fig:fig4-heavyhex} the infidelity attained on the $N=30$ heavy-hex lattice fragment as a function of two-qubit gates applied to prepare a kicked Ising target $(\theta_h=0.4)$ with the MPD and ASPIRE is shown. We model three noise levels $p=\{10^{-3}, 2\times10^{-4}, 10^{-4}\}$ (colour coded light to dark) alongside the noiseless case (dashed) with the MPD in blue and ASPIRE in orange. The lowest of these is at or beyond the best reported two-qubit error rates, and is a projection rather than a present-day operating point. To account for the average SWAP costs of the MPD we charge it $1.4\times$ its logical gate-count, and compare on the axis of native two-qubit gates applied.  In all three noise levels ASPIRE is able to reach both a lower infidelity floor and to get there in fewer gates, with weaker levels of noise leading to larger gate-count discrepancies and more gates to reach the minima. Panel \textbf{(b)} tests the noise-scaled selectivity proposed in Section~\ref{sec:conclusion} on the same patch. The channel is applied as a global depolarising channel, the only model an $N=30$ register allows. Appendix~\ref{app:noise} gives its closed form and checks it against exact density-matrix simulation of hardware-learned noise on a $12$-qubit cell (Figure~\ref{fig:noise-model-check}), where the global model is the pessimistic one for both protocols and changes the ratio of their best infidelities by at most $0.07$.

\begin{figure}
    \centering
    \includegraphics[width=\linewidth]{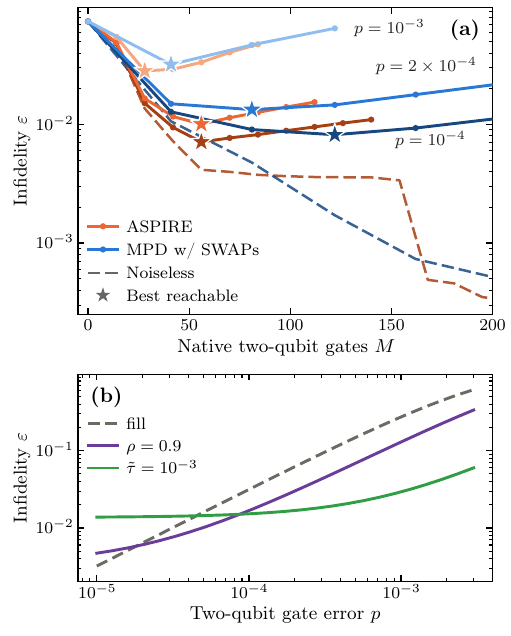}
    \caption{ \textbf{Heavy-hex kicked-Ising state preparation with gate noise.} \textbf{(a)} Infidelity against native two-qubit gates for a kicked-Ising target state ($\theta_h=0.4$) prepared on an $N=30$ heavy-hex fragment (three hexagons) at three gate error rates, each native gate applying a depolarising channel of rate  $p=\{10^{-3}, 2\times10^{-4}, 10^{-4}\}$ (colour coded light to dark), with the noiseless case shown in dashed lines. As in Figure~\ref{fig:overview}\textbf{(d)}, the MPS mapping leaves the MPD with chain-adjacent pairs the device does not couple and it can only realise them through SWAP sequences, whereas ASPIRE selects only from hardware-native gates. The MPD is charged $1.4 \times$ its logical gate count to account for the SWAPs, and every curve descends whilst additional gates remove more correlation than it injects error. Eventually an infidelity floor for a given noise level is reached and the infidelity starts rising, with the best reachable point marked with a star. \textbf{(b)} Infidelity at which each ASPIRE admission rule stops, against gate error $p$, on the same patch and noise model as \textbf{(a)} for three admission rules: fill (dashed, grey) takes every pair in the maximum-weight matching, relative $\rho=0.9$ (solid, purple) keeps a pair only if it scores $\geq0.9$ in its unrestricted layer, absolute $\tilde \tau=10^{-3}$ (solid, green)  keeps a pair only if it clears an absolute floor of $10^{-3}$ bits.}
    \label{fig:fig4-heavyhex}
\end{figure}

\subsection{Multivariate functions}
\label{subsec:quantics}

The Hamiltonian ground states studied thus far have an entanglement geometry inherited from the shape of their 2D lattice. However, matrix product states have recently proven to be an extremely versatile tool for representing and manipulating amplitude-encoded functions~\cite{holmes2020efficient,garcia2021quantum,gourianov2022quantum,ye2022quantum,shinaoka2023multiscale,melnikov2023quantum,chen2023quantum,iaconis2024quantum,bohun2026entanglement,lee2026quantum,connor2026tensor,niedermeier2026solving,sims2026multivariate}, thanks to provably low bond dimension for well-behaved functions~\cite{oseledets2010approximation,khoromskij2011quantics,holmes2020entanglement,bohun2026entanglement} and rank-revealing constructive methods based on tensor cross-interpolation (TCI)~\cite{nunez2022learning,ritter2024quantics,nunez2025learning}; this representation is typically referred to as the `quantics' representation in tensor-network literature. In the simplest case of a univariate function $f:[0,1) \to \mathbb{R}$, the function is discretised on a uniform grid of size $2^N$ at grid points $x_j = j/2^N$ for $j=0,\dots,2^N-1$, and encoded as the state
\begin{equation}
  \ket{\psi}=\sum_j f(x_j) \ket{j}, \quad x_j \equiv \sum_{b=1}^{N}\frac{s_{b}}{2^b} = \underbrace{0.s_{1} s_{2} \dots s_{N}}_{\text{binary representation}},\label{eqn:quantics_representation}
\end{equation}
where the basis states $\ket{j} \equiv \ket{s_{1} s_{2} \dots s_{N}}$ index the grid points in binary counting order.
A general bounded domain can be accommodated by a simple linear rescaling, and multivariate functions can be accommodated via the product of several such registers: a $D$-dimensional function $f(x_1,\dots,x_D)$ discretised on a $2^B$-point grid in each dimension, is encoded using one $B$-qubit register per variable as
\begin{equation}
\ket{\psi} = \sum_{j_1,\dots,j_D} f(x_{j_1},\dots,x_{j_D}) \ket{j_1}\otimes\cdots\otimes\ket{j_D} \in \mathcal{H}^{\otimes BD}.\label{eqn:multivariate_psi}
\end{equation}

Varying the number of variables $D$ or bits of precision $B$ offers two different ways of increasing the problem size, which may impact ASPIRE and MPD methods differently. The entanglement geometry of the resulting MPS will clearly depend on how the bits are ordered, which may vary depending on whether one prioritises minimising the range of intra- or inter-variable qubit correlations --- the standard orderings in the literature are the \emph{interleaved} and \emph{serial} representations~\cite{gourianov2022quantum,ye2022quantum,shinaoka2023multiscale,ritter2024quantics,nunez2025learning,lee2026quantum}. For variables $x,y,...$ the serial ordering $s^x_1, \cdots, s^x_B, s^y_1, \cdots, s^y_B$ keeps each variable's qubits contiguous, minimising intra-variable correlation range and placing cross-variable correlation at a range of ${\approx}B$ sites, e.g. for $D=3$ the MPS takes the form
\begin{equation}
\bra{\bm{s}^x,\bm{s}^y,\bm{s}^z}\ket{\psi} =
\begin{array}{c}
\begin{tikzpicture}
  \Vertex[x=0,   y=0,position=above,size=0.25,color=blue!30]{A11}
  \Vertex[x=1.2, y=0,position=above,size=0.25,color=blue!30]{A1N}
  \Vertex[x=2.1,y=0,position=above,size=0.25,color=orange!60]{A21}
  \Vertex[x=3.3,y=0,position=above,size=0.25,color=orange!60]{A2N}
  \Vertex[x=4.2, y=0,position=above,size=0.25,color=green!60]{A31}
  \Vertex[x=5.4, y=0,position=above,size=0.25,color=green!60]{A3N}

  % physical legs
  \draw[line width=1] (A11) -- ++(0,-0.35) node[below] {\scriptsize$s^x_1$};
  \draw[line width=1] (A1N) -- ++(0,-0.35) node[below] {\scriptsize$s^x_B$};
  \draw[line width=1] (A21) -- ++(0,-0.35) node[below] {\scriptsize$s^y_1$};
  \draw[line width=1] (A2N) -- ++(0,-0.35) node[below] {\scriptsize$s^y_B$};
  \draw[line width=1] (A31) -- ++(0,-0.35) node[below] {\scriptsize$s^z_1$};
  \draw[line width=1] (A3N) -- ++(0,-0.35) node[below] {\scriptsize$s^z_B$};

  % within-register bonds: dashed, standing in for the N-2 hidden tensors
  \Edge[style={dashed},label={\makebox[1em][c]{$\cdots$}}](A11)(A1N)
  \Edge[style={dashed},label={\makebox[1em][c]{$\cdots$}}](A21)(A2N)
  \Edge[style={dashed},label={\makebox[1em][c]{$\cdots$}}](A31)(A3N)

  % between-register bonds: solid, ordinary single bonds -- this is the "serial ordering" point
  \Edge[lw=1,position=below](A1N)(A21)
  \Edge[lw=1,position=below](A2N)(A31)
\end{tikzpicture}
\end{array},
\label{eqn:serial_quantics_mps}
\end{equation}
where the dashed lines denote the $B-2$ middle tensors of each variable, notationally suppressed.
In contrast, the interleaved ordering $s^x_1,s^y_1,\cdots,s^x_2,s^y_2,\cdots$ minimises inter-variable correlation by interleaving qubits of different variables at the same bit significance, at the cost of placing intra-variable correlation at distance ${\approx}D$ sites, e.g. for $D=3$
\begin{equation}
\bra{\bm{s}^x,\bm{s}^y,\bm{s}^z}\ket{\psi} =
\begin{array}{c}
\begin{tikzpicture}
  % first significance level: x, y, z bits adjacent
  \Vertex[x=0,   y=0,position=above,size=0.25,color=blue!30]{X1}
  \Vertex[x=0.9, y=0,position=above,size=0.25,color=orange!60]{Y1}
  \Vertex[x=1.8, y=0,position=above,size=0.25,color=green!60]{Z1}
  % last significance level
  \Vertex[x=3.0, y=0,position=above,size=0.25,color=blue!30]{XB}
  \Vertex[x=3.9, y=0,position=above,size=0.25,color=orange!60]{YB}
  \Vertex[x=4.8, y=0,position=above,size=0.25,color=green!60]{ZB}
  % physical legs
  \draw[line width=1] (X1) -- ++(0,-0.35) node[below] {\scriptsize$s^x_1$};
  \draw[line width=1] (Y1) -- ++(0,-0.35) node[below] {\scriptsize$s^y_1$};
  \draw[line width=1] (Z1) -- ++(0,-0.35) node[below] {\scriptsize$s^z_1$};
  \draw[line width=1] (XB) -- ++(0,-0.35) node[below] {\scriptsize$s^x_B$};
  \draw[line width=1] (YB) -- ++(0,-0.35) node[below] {\scriptsize$s^y_B$};
  \draw[line width=1] (ZB) -- ++(0,-0.35) node[below] {\scriptsize$s^z_B$};
  % bonds within a significance level: solid
  \Edge[lw=1,position=below](X1)(Y1)
  \Edge[lw=1,position=below](Y1)(Z1)
  \Edge[lw=1,position=below](XB)(YB)
  \Edge[lw=1,position=below](YB)(ZB)
  % dashed: the 3(B-2) suppressed tensors of significance levels 2,...,B-1
  \Edge[style={dashed},label={\makebox[1em][c]{$\cdots$}}](Z1)(XB)
\end{tikzpicture}
\end{array},
\label{eqn:interleaved_quantics_mps}
\end{equation}
where the dashed line now represents the suppressed $D(B-2)$ tensors for bits of significance levels $2$ through $B-1$. The interleaved representation typically leads to a more efficient representation for non-separable functions~\cite{ye2022quantum,sims2026multivariate}, while the serial ordering is popular for the ease of applying MPOs acting on a single variable contiguously~\cite{lee2026quantum}.

This distinction between representations is important here since they represent the same state at the same $N$ with different correlation ranges. This can present a severe bottleneck for local, sequential methods based on the MPD, since the locality of correlations is entirely representation-dependent even though both representations may be efficient. In contrast, ASPIRE relies only upon pairwise entanglement metrics, and only `sees' the true underlying entanglement geometry rather than this artefact of the MPS representation; its behaviour is identical on either representation (up to truncation error) when all-to-all connectivity is permitted.

\begin{figure*}
    \centering
    \includegraphics[width=\linewidth]{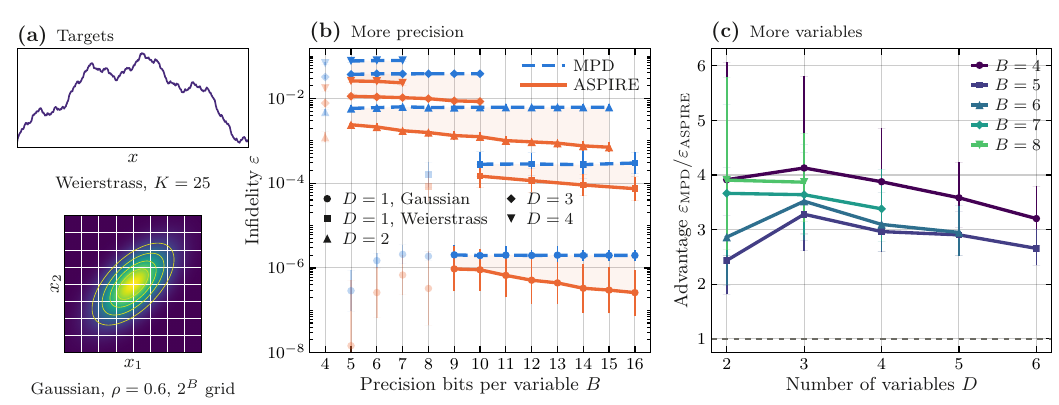}
       \caption{ \textbf{Quantics function states prepared with a matched gate budget for MPD and ASPIRE.} Infidelity of the MPD (blue) and ASPIRE (orange) in the quantics encoding with serial bit ordering, gate-matched at the MPD's $8$-layer budget of $8(N-1)$ two-qubit gates on $N=DB$ qubits, each arm read at its lowest infidelity within that budget. Error bars are the seed spread and the shaded band is the separation. The Gaussian lines with $D\ge2$ use $16$ seeds and maximum bond dimension $256$, and the two $D=1$ lines use $8$ seeds and $1024$; no cap binds. A seed draws a member of the function family (random correlation signs, widths and centres for the Gaussian; random amplitudes and phases for the Weierstrass), so the spread is over targets and not over algorithmic randomness. \textbf{(a)} The two target families: the Weierstrass series $\sum_{k=1}^{K}\sin(\pi k^{3}x)/(\pi k^{3})$ of Ref.~\cite{tindall2024compressing}, drawn at $K=25$ with unit amplitudes and zero phases, and the equicorrelated Gaussian ($\rho=0.6$) on its own domain with the $2^B$ sampling grid overlaid. \textbf{(b)} Refining the grid at fixed $D$; the panel also carries a $D=1$ Weierstrass series read off a separate budget run which differs from its neighbours in both bond dimension and function and is labelled for that reason. Faded, unjoined markers are the small-$B$ points before ASPIRE's infidelity peaks, where it still rises with $B$, and no trend is drawn through them. On the Gaussian lines the MPD's error is flat in $B$ because the target's best bond-dimension-two approximation stops moving once the grid is fine enough, while the cross-variable correlation the added bits carry lies beyond the reach of nearest-neighbour gates; ASPIRE converts the growing budget into fidelity and the separation widens. \textbf{(c)} Adding variables at fixed $B$, the axis on which the target hardens for both protocols. The MPD/ASPIRE infidelity ratio is plotted directly on a linear axis, one line per precision $B=4,\dots,8$; it is flat or rising from $D=2$ to $D=3$ and decreases beyond.}
    \label{fig:fig5-quantics}
\end{figure*}

In the serial ordering, at fixed variable count $D$, refining the precision of the grid by increasing $B$ leads to growing inter-variable correlations at length ${\approx}B$ sites. The MPD cannot place gates directly targeting this growing correlation length, since it can only place sequential gates and the inter-variable correlations are non-local, whilst ASPIRE is insensitive to site ordering. Consequently, one would expect that ASPIRE would increasingly outperform the MPD with increasing $B$ for a serial quantics MPS, which is precisely what we observe in Figure~\ref{fig:fig5-quantics}\textbf{(b)}. We evaluate this further in Appendix~\ref{app:quantics}. This insensitivity has a counterpart outside the quantics family. In quantum-chemistry DMRG the orbital ordering is routinely optimised against a mutual-information measure before the calculation is run~\cite{legeza2003optimizing, rissler2006measuring}, precisely because the accuracy attainable at fixed bond dimension depends on it. By Section~\ref{sec:placement} our selection rule is invariant under that reordering, so the ordering search such work performs has no analogue in our protocol. While an efficient ordering is still needed to ensure the underlying tensor-network operations are tractable, our protocol improves upon the MPD by ensuring that the ordering has no consequences for the outcome of circuit compilation.

Figure~\ref{fig:fig5-quantics} separates the variable and precision axes at a two-qubit budget matched to the MPD's eight-layer circuit construction. Refining the grid at fixed $D$ (panel \textbf{(b)}) leaves the MPD's error essentially unmoved, for the bivariate Gaussian it is pinned near $6.2\times10^{-3}$ from $B=6$ onwards, while ASPIRE continues to descend, so the separation widens from $3.9\times$ at $B=4$ to $8.8\times$ at $B=15$. Two things make MPD's infidelity flat in $B$. Firstly the target's best bond-dimension-two approximation error stops moving once the grid is fine enough, measured at $2.85\times10^{-2}$ for the bivariate Gaussian and constant to $0.2\%$ over $B=6$ to $10$, and the MPD's first layer realises exactly that approximation. Secondly, the added precision bits introduce cross-variable correlation in their bits that the MPD cannot discern, whilst not affecting the intravariable correlation that it targets effectively with its layer count. ASPIRE, in contrast, uses the added precision bits as a wider pool of target pairs and is blind to whether a correlation is between variables or not. In \textbf{(c)} the advantage factor is plotted directly against the number of variables. The target hardens for both protocols as $D$ grows, and the advantage factor decreases with more variables. We note that this comparison was performed with a fixed number of layers, such that the infidelities rise as high as $O(0.1)$ which may influence the comparison. Ultimately, ASPIRE attains a lower infidelity than the MPD for all quantics functions considered, and indeed seems to scale better with more precision bits per variable than the MPD.

Panel \textbf{(a)} of Figure~\ref{fig:fig5-quantics} shows the two target families. The Weierstrass target is the truncated series $\sum_{k=1}^{K}a_k\sin(\pi k^{3}x+\phi_k)/(\pi k^{3})$ of Ref.~\cite{tindall2024compressing}, whose frequencies grow as $k^{3}$. Each seed draws the amplitudes $a_k$ and phases $\phi_k$, and the number of terms grows with the grid, $K=\lfloor(2^{B}/4)^{1/3}\rfloor$, so that the finest mode is always sampled about four times per period ($K=25$ at $B=16$, the member drawn in panel \textbf{(a)} with $a_k=1$, $\phi_k=0$). Each refinement of the grid therefore adds finer modes, so every added precision bit carries genuine weight, yet the MPD's error on the $D=1$ Weierstrass line of \textbf{(b)} stays flat. The equicorrelated Gaussian below is the opposite regime: it is smooth, its bond spectrum decays quickly, and the correlation the added bits carry is almost entirely between variables, sitting at a range of ${\approx}B$ sites in the serial ordering. The $2^B$ grid is overlaid to show the discretisation; it is the serial encoding, not the function, that places this correlation far apart on the chain. Reading panel \textbf{(a)} alongside \textbf{(b)} shows how the MPD's nearest-neighbour layers reach the intra-variable chain but not the Gaussian's cross-variable correlation, whereas ASPIRE selects pairs by the correlation itself and so is indifferent to how the encoding lays the variables out.

\section{Conclusions} \label{sec:conclusion}

In this work, we introduced ASPIRE, a protocol to produce approximate state-preparation circuits tailored to both the state's underlying entanglement geometry and the available connectivity on the target quantum processor. The algorithm considers the space of all possible two-qubit gates, and uses entanglement metrics to select the gates that will remove the maximal amount of correlation from the working state. The protocol is natively adaptive due to the parameterisation of its entanglement scoring, gate selectivity, gate generation, and layer size. We have demonstrated that it can prepare higher fidelity states than other approximate algorithms at matched gate count when there is a long-ranged entanglement structure to discover, persisting when classical optimisation is applied and generally being more robust to noise.  

\subsection{Tensor-network topologies}\label{sec:tensor_network_topologies}
In this work, we have demonstrated how building state preparation around the long-ranged entanglement geometry of the target state can achieve superior performance to MPD methods, to an extent that improves with increasing native connectivity in the hardware. In all demonstrations in this work we have produced state preparation circuits for MPS, since this is a popular target in the literature that is well-motivated for a broad range of interesting problems, even beyond 1D~\cite{chan2011density,baiardi2020density,menczer2024tensor,menczer2024parallel,xiang2024distributed,legeza2025orbital,white2007neel,yan2011spin,stoudenmire2012studying,zheng2017stripe,zhou2020limits,garcia2021quantum,ye2022quantum,gourianov2022quantum,shinaoka2023multiscale,ritter2024quantics,nunez2025learning,bohun2026entanglement,sims2026multivariate,lee2026quantum,connor2026tensor,niedermeier2026solving}, and the most straightforward tensor-network representation to prototype new methods. However, our protocol is not inherently limited to MPS --- indeed, much of the strength of our protocol comes from the fact that ASPIRE does not `see' the underlying 1D geometry of the representation, but only the entanglement geometry of the state itself, allowing us to tailor our state-preparation circuits to long-ranged correlations beyond 1D. Since ASPIRE does not rely upon geometric properties of the underlying classical representation, there is no reason to expect that it cannot be generalised to other underlying classical representations. Extending state-preparation methods to other tensor-network topologies has greatly expanded the range of efficiently preparable states on quantum computers~\cite{haghshenas2022variational,zhang2024characterizing,sugawara2025embedding,sims2026multivariate,szoldra2026belief,mingare2026practical}, and this principle might extend even beyond tensor networks.

For a given classical representation of a state, to perform the ASPIRE protocol we must be able to (i) efficiently compute the pairwise QMI, and (ii) efficiently apply the discovered gates and their inverses to the state's representation. For more general isometric tensor networks like tree tensor networks~\cite{shi2006classical,sugawara2025embedding,sims2026multivariate}, MERA~\cite{vidal2007entanglement,vidal2008class,evenbly2009algorithms}, and isoPEPS~\cite{haghshenas2019conversion,zaletel2020isometric} (i) is straightforwardly achieved by standard contraction algorithms, while (ii) might be achieved with e.g. penetration algorithms~\cite{seitz2023simulating,sugawara2025embedding} or Evenbly-Vidal-like reoptimisation~\cite{evenbly2009algorithms}. There is also scope to extend beyond tensor networks --- for example, (i) would be straightforward with fermionic Gaussian states~\cite{peschel2003calculation} or stabiliser states~\cite{fattal2004entanglement}, and (ii) could also be achieved within the restricted set of fermionic Gaussian or Clifford unitaries (which are guaranteed to be sufficient to prepare the state). We suggest that recent advances combining Clifford and tensor-network methods~\cite{lami2024unveiling,qian2024augmenting,nakhl2025stabilizer,masot2026limits} may provide further avenues for enhanced state preparation following this simple principle of entanglement-selected gate placement.

\subsection{Mechanism: reach versus objective}

ASPIRE differs from MPD both in which pairs of sites a gate may act on, and what each gate optimises. While MPD maximises overlap with the $\chi=2$ truncation, Algorithm~\ref{alg:state-prep} maximally reduces the entropy of two single-site marginals. Thus our algorithm uses a local objective that is computable for any two sites, whereas MPD uses a global objective that is restricted to pairs adjacent in the MPS. At $r=1$ the two protocols act on the same connections, and we generally find MPD to be more performant in some regimes (Figure~\ref{fig:fig3-connectivity-spectra}), and tied in others (Figure~\ref{fig:decomposition}), although certain systems are still more suited to ASPIRE's objective (Figure~\ref{fig:toy-expressivity}). For the spin systems, once the radius covers the longest lattice bond in the chain ordering, ASPIRE covers most of the entanglement geometry and consistently outperforms MPD. On random matrix product states, which carry no such scale, the two protocols nearly tie (Figure~\ref{fig:fig2-randomMPS}\textbf{(c)}), which is the prediction the entanglement geometry makes and not a shortfall of the method. A tiebreaker in the controlled example can for example be the SWAP costs. If the random MPS is prepared on a platform which does not connect MPD's bonds natively then the added SWAP costs would push the advantage in ASPIRE's favour, whereas if the hardware has too little connectivity ASPIRE may need SWAP gates for its long ranged pairs which would have the inverse effect.

\subsection{Implications for the two regimes}

In NISQ, the ASPIRE protocol comes with the clear advantage arising from the selectivity at which gates are applied. Whereas MPD requires the full brickwork structure to iteratively prepare a state, ASPIRE can synthesise its layers in batches ranging from a single high-impact gate to a full maximum-weight matching. Therefore, we propose scaling the selectivity of the gates applied to the level of noise present in the system, such that a higher degree of noise is accounted for by running shallower circuits with higher-impact gates. A noise-aware protocol overlooks part of the entanglement geometry, but by avoiding gates that inject more error than they remove it can reach lower infidelity in the NISQ regime. In Figure~\ref{fig:fig4-heavyhex} \textbf{(b)} we test this by comparing ASPIRE run with no gate-selectivity (dashed, grey) to a relative selectivity requiring a gate to score $\geq0.9$ in its unrestricted layer (solid, purple) and an absolute selectivity requiring gates to remove $\geq3\times10^{-3}$ bits to be accepted (solid, green). Both protocols are found to outperform the naive implementation, reaching lower infidelities at the end of their Algorithm~\ref{alg:state-prep} past a certain noise level $p=O(10^{-5})$ with the improvement generally increasing with stronger noise, although flattening out and plausibly decreasing slightly at high noise. Whilst this regime is characterised by still limited connectivity,  Figure~\ref{fig:fig3-connectivity-spectra}, Figure~\ref{fig:fig3-connectivity-trends}, and Figure~\ref{fig:fig4-heavyhex} indicate that some-to-some connectivity is sufficient to achieve an advantage over MPD.

In the early fault-tolerant regime the $T$ count is a key metric of comparison, insensitive to the range of the gates used. In Figure~\ref{fig:fig1-efficacy} ASPIRE reaches $\varepsilon=2.13\times10^{-2}$ at $2.8\times10^{5}$ $T$ gates after a single optimisation sweep, and $1.61\times10^{-2}$ at $7.0\times10^{5}$ without optimisation, whereas the MPD, optimised at the same classical budget, does worse at $3.05\times10^{-2}$ at $5.8\times10^{5}$ $T$ gates. QPE is a key application of MPS state preparation, and the accuracies reached at these manageable T-gate counts are likely already sufficient for it. Pushing fidelity further with exact techniques costs significantly more quantum resources, yet can only raise the QPE success probability by a small amount. Figure~\ref{fig:qpe-triangle} shows that the triangle inequality on the Fubini--Study angle turns the measured infidelity $\varepsilon$ and the DMRG reference quality $p$ into a guaranteed lower bound on the true QPE success probability. The bound tightens quickly as $\varepsilon\to0$: the marked points show where the guarantee comes within $0.01$ of its ceiling $p$, which for $p=0.99$ already occurs at $\varepsilon\approx1.7\times10^{-3}$. Circuit infidelity pushed well below the DMRG error $1-p$ buys little, since the guarantee saturates at $p$ regardless of $\varepsilon$.

\begin{figure}
    \centering
    \includegraphics[width=\linewidth]{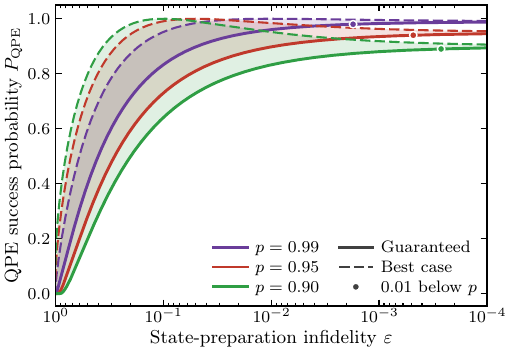}
    \caption{\textbf{QPE success probability vs.\ state-prep infidelity.} $P_{\mathrm{QPE}}(\mathrm{success})=|\langle a|c\rangle|^2$, the overlap of the prepared circuit state $\ket{c}$ with the true target $\ket{a}$, bounded from the measured infidelity $\varepsilon=1-|\langle b|c\rangle|^2$ between $\ket{c}$ and the DMRG/MPS reference $\ket{b}$, via the Fubini--Study triangle inequality at DMRG quality $p=|\langle a|b\rangle|^2\in\{0.99,0.95,0.90\}$ (solid: guaranteed lower bound; dashed: best-case upper bound).}
    \label{fig:qpe-triangle}
\end{figure}

\subsection{Limitations}
\label{sec:limitations}
We also note that the ASPIRE protocol, while powerful, is fundamentally limited by the underlying classical representation of the state. While MPS have been a powerful tool for studying systems beyond 1D in chemistry~\cite{chan2011density,baiardi2020density,menczer2024tensor,menczer2024parallel,xiang2024distributed,legeza2025orbital}, condensed matter physics~\cite{white2007neel,yan2011spin,stoudenmire2012studying,zheng2017stripe} and the study of quantum circuits~\cite{zhou2020limits}, only 1D area-law states can be expected to scale efficiently indefinitely~\cite{eisert2010colloquium}. While the aforementioned application areas show that MPS is sufficient to represent initial states for interesting utility-scale problems, this will not be the case for every system at every scale. The viability of e.g. initialising QPE for general chemistry problems remains in contention~\cite{lee2023evaluating}, and the limitation applies to all state-preparation schemes that start from a classical approximation. Nonetheless, as outlined in Section~\ref{sec:tensor_network_topologies}, we anticipate that expanding the simple principle of entanglement-geometry-guided circuit design to other classical representations, both more advanced tensor-network topologies and non-tensor-network methods, will further broaden the scope of systems to which these methods are applicable.

We also note that on a stabiliser state whose minimum stabiliser weight is
$w\ge3$ every pair scores exactly zero (Appendix~\ref{app:stab}), so the protocol emits no
gate at all. This problem can be circumvented by implementing composite moves of multiple two-qubit gates determined from a nonzero composite score. The composite moves (discussed in Appendix~\ref{app:composite}) repair the
stall on stabiliser targets at two native gates rather than the much higher number compiling a
three-qubit unitary would cost.  This stall is not a practical limitation, since stabiliser states are prepared efficiently by Clifford synthesis \cite{aaronson2004improved}.

Outside of the stabiliser case, a pair can carry $I(i:j)>0$ at $R_{ij}=0$
(Corollary~\ref{cor:sat}), and this occurs routinely. It is not, however, a
deficiency of the score: $R_{ij}$ is the largest reduction in
$S(\rho_i)+S(\rho_j)$ available to \emph{any} two-qubit unitary on that pair, so
no alternative pair metric and no composite of pair metrics can make such
a pair actionable. This limitation is also bounded: the mutual information left on a pair by its gate
depends only on the spectrum of $\rho^{ij}$, and is small in practice. Since $R_{ij}=0$ is a property of the pair in the current state rather than of the pair itself, a gate on an overlapping support can restore $R_{ij}>0$, and in our numerics frozen pairs usually became actionable again within one layer (Appendix~\ref{app:gate}). This recovery is an empirical observation, not a guarantee. Every admissible pair can score zero at once while correlation remains on the register, and stabiliser states are not the only such targets: Appendix~\ref{app:composite} gives a seven-qubit non-stabiliser state on which every pair carries $I(i{:}j)\approx0.062$ bits at $R_{ij}=0$, so that Algorithm~\ref{alg:aspire} halts without emitting a gate. We found no such stall on the targets of Section~\ref{sec:results}, and it is removed by a simple fallback that none of the results reported here required: when no candidate survives step~(II) before the budget is spent, one MPD layer (Algorithm~\ref{alg:mpd}) is applied to the working state and selection resumes. Product states are among the bond-dimension-two states that this layer prepares, so with an overlap-optimal truncation it cannot lower the attainable fidelity, and on the example of Appendix~\ref{app:composite} a single layer restores pair scores of up to $0.65$ bits.

A further limitation concerns fermionic encodings. Our selection rule
reads the entanglement geometry of the encoded qubit state, and a fermionic state
acquires such a geometry only once an encoding fixes one. The Jordan--Wigner
encoding~\cite{jordan1928uber,ortiz2001quantum} maps the occupation basis to the
computational basis, so the single-site marginals the rule reads are the fermionic
mode marginals themselves, and $S(\rho_i)$ and $I(i{:}j)$ reduce to the single-orbital
entropy and orbital mutual information already used to optimise orbital orderings in
quantum-chemistry DMRG~\cite{legeza2003optimizing,rissler2006measuring}. Encodings such as parity~\cite{seeley2012bravyi}
or Bravyi--Kitaev~\cite{bravyi2002fermionic}
differ by an entangling Clifford network rather than by a permutation of sites ~\cite{mcardle2020}, and
therefore present the protocol with different single-site marginals, different
$R_{ij}$, and hence different circuits. Consequently,
should the protocol be applied to a fermionic target, whatever advantage it achieves
would be a property of the pairing (target, encoding) rather than of the target alone,
and we take Jordan--Wigner to be the natural choice for a geometry-driven selection
rule; we have not measured how the advantage would change under a different encoding.

\section*{Acknowledgements}
We acknowledge helpful technical discussions with Tom Bromley, Hamza Jnane, Jona Erle, and B\'{a}lint Koczor. The authors acknowledge HPC provisions funded by the UKRI Future Leaders Fellowship award Theory to Enable Practical Quantum Advantage (MR/Y015843/1).

The project concept was scoped and conceived by M.L.S-G. and refined by F.H. Implementation and numerical experiments were conducted by F.H. and supervised by M.L.S-G. The code for the numerics was built on a fork of \texttt{MPSCircuits.jl}~\cite{mpscircuits} (a codebase authored by M.L.S-G.), with further development conducted by F.H. assisted by frontier LLMs under close supervision. The code accompanying this paper can be found in the GitHub repository of Ref. \cite{ASPIRE}. All text in this manuscript is principally authored by F.H. and M.L.S-G., with LLM tools used for feedback, formatting, and miscellaneous assistance.

\bibliography{references}

\clearpage
\newpage
\widetext

\appendix
\section*{Appendices}
\section{Noise models}
\label{app:noise}
The noisy results of Section~\ref{subsec:noise} use a global-depolarising noise model, chosen because its impact on the fidelity can be evaluated in closed form at system sizes too large for direct density-matrix simulation, and since its analytical tractability allows its fidelity impact to be evaluated in closed form for system sizes too large for direct density-matrix simulation. Global depolarising noise is a reasonably well-motivated idealisation, since sufficiently deep random circuits provably scramble weak local noise into global depolarising noise~\cite{dalzell2021random}. However, realistic hardware noise is local and can be strongly biased across Pauli errors, so we verify this model against density-matrix simulations of a more realistic hardware noise model at smaller system sizes. We detail the global depolarising model in Section~\ref{app:global_depol} and the realistic noise model in Section~\ref{app:realistic_noise_model}. Throughout, each two-qubit gate $G_n\in\operatorname{SU}(4)$ is assigned an error probability $p$, single-qubit gates are assumed noiseless, and noise is simulated only on evaluation of completed circuits, with gate selection and optimisation always using the noiseless circuit computed via MPS methods.

\subsection{Global depolarising noise}\label{app:global_depol}
After each two-qubit gate, noise is applied to the density operator $\rho$ according to the global depolarising channel
\begin{equation}
  \rho \mapsto \left(1-\frac{16p}{15}\right)\rho + \frac{16p}{15}\,\frac{\mathds{1}}{2^N},
  \label{eq:global-depol}
\end{equation}
where the weight $16p/15$ matches the convention for a two-qubit depolarising channel of error probability $p$; this differs from a local two-qubit depolarising channel in that it scrambles the entire register rather than a single pair, and also commutes with every gate such that a circuit of $M$ gates produces the state
\begin{equation}
  \rho_{\text{out}} = q\dyad{\psi_{\text{out}}} + \left(1-\left(1-\frac{16p}{15}\right)^G\right)\frac{\mathds{1}}{2^N}
\end{equation}
where $\ket{\psi_{\text{out}}}$ is the noiseless output state of the circuit.
The noisy fidelity therefore follows directly from the noiseless fidelity as
\begin{equation}
\bra{\psi_{\text{target}}}\rho\ket{\psi_{\text{target}}} = qF_0 + \frac{1-q}{2^N},
\end{equation}
where $F_0=|\braket{\psi_{\text{target}}}{\psi_C}|^2$ is computed by MPS-MPS contraction~\cite{schollwock2011density}, and we report $\varepsilon = 1-\sqrt{\bra{\psi_{\text{target}}}\rho\ket{\psi_{\text{target}}}}$, which reduces to Equation~\ref{eq:infidelity} at $p=0$.

\subsection{Biased noise model}\label{app:realistic_noise_model}
For exact density-matrix simulations we use the more realistic noise model outlined in the appendices of Refs.~\cite{ijaz2025more,goh2026direct}. Following each two-qubit gate, the pair it acted on is subjected to the two-qubit local Pauli channel
\begin{equation}
  \Lambda(\rho) = \sum_{j=0}^{15} c_j P_j \rho P_j, \qquad p \equiv 1-c_0,
  \label{eq:pauli-channel}
\end{equation}
where the amplitudes $c_j$ are generated from a sparse Pauli-Lindblad model~\cite{van2023probabilistic} whose coefficients are sampled from those learned from a real quantum processor; this approach is sufficiently realistic that it has enabled highly effective error mitigation on real devices~\cite{van2023probabilistic,kim2023evidence}. While a local-depolarising channel is unbiased over the 15 non-identity 2-qubit Pauli operators, this model is heavily biased towards certain error channels with a bias that differs inhomogeneously for each qubit pairing. This better captures the character of real hardware noise channels, both in local bias and inhomogeneity across the processor.

\subsection{Verification on a 12-qubit cell}
We compare the three models on one hexagon of the heavy-hex lattice, $N=12$ qubits on a $12$-cycle of couplers, the largest fragment we simulate exactly. Because this cell is a ring, every MPS bond is a coupler, neither protocol needs routing, and the noise channel is the only thing that differs between the models. The target is a kicked transverse-field Ising state from the same family as Figure~\ref{fig:fig4-heavyhex}, at $\theta_h=0.6$. We construct one set of noiseless circuits per protocol, six MPD layers ($66$ gates) and an ASPIRE circuit restricted to the cell's couplers at the same gate count, and score every checkpoint three ways: with the closed form of Equation~\ref{eq:global-depol}, and by evolving the full density matrix in QuEST~\cite{jones2019quest} (via the pyQuEST interface) under local depolarising and under sparse Pauli--Lindblad noise. As consistency checks, the density-matrix simulation reproduces the tensor-network fidelity of every circuit at $p=0$, and setting all Pauli--Lindblad rates equal reproduces local depolarising noise to within $10^{-12}$.

Figure~\ref{fig:noise-model-check} supports two conclusions. First, global depolarising noise is the pessimistic model, as anticipated: the exact curves have the same shape as the global ones and reach their minimum at the same checkpoint or one MPD layer later, but lie below them, with the best reachable infidelity $0.69$--$0.86$ times its global value under sparse Pauli--Lindblad noise and $0.80$--$0.90$ times under local depolarising noise, for both protocols. Second, and more importantly, the comparison between the protocols does not depend on the model: the ratio of the MPD's best infidelity to ASPIRE's changes by at most $0.07$ between models, in no consistent direction. This check isolates the noise channel. Since neither protocol routes on this cell, the two are within $16\%$ of each other under every model; the larger separation in Figure~\ref{fig:fig4-heavyhex} comes from the SWAP overhead the MPD incurs on a larger device, which is set by Equation~\ref{eq:routing-overhead} rather than by the choice of noise model.

\begin{figure*}
    \centering
    \includegraphics[width=\linewidth]{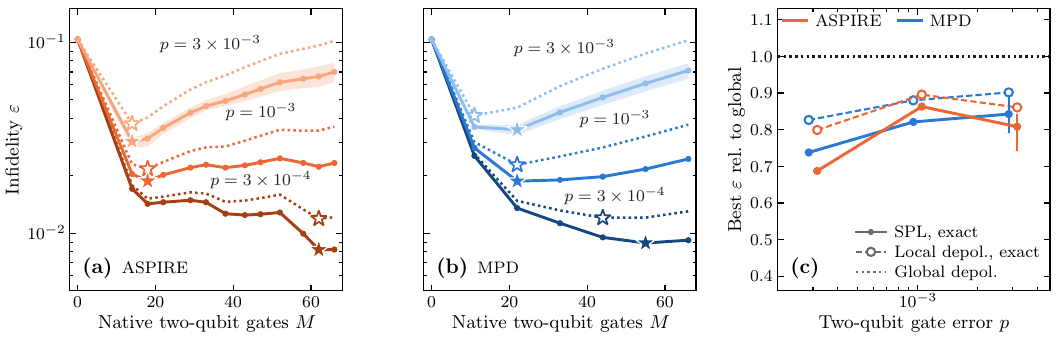}
    \caption{\textbf{Global depolarising noise against exact simulation of hardware-learned noise.} Infidelity against native two-qubit gates for \textbf{(a)} ASPIRE and \textbf{(b)} the MPD, preparing the kicked transverse-field Ising state at $\theta_h=0.6$ on one hexagon of the heavy-hex lattice ($N=12$ qubits on a $12$-cycle, so every MPS bond is a coupler and neither protocol routes), at three gate error rates $p=\{3\times10^{-4},10^{-3},3\times10^{-3}\}$ (colour coded dark to light). Solid lines evolve the full density matrix under the sparse Pauli--Lindblad channel after every two-qubit gate, with each coupler's own rates rescaled to a total error of exactly $p$. At $p=3\times10^{-3}$ the solid line is the mean of five rate draws and the band their range, and at the two weaker rates it is a single draw. Dotted lines score the same checkpoint circuits with the global depolarising closed form used for Figure~\ref{fig:fig4-heavyhex}. Filled stars mark the minimum of each solid curve and open stars that of each dotted curve. \textbf{(c)} Best reachable infidelity under the exact sparse Pauli--Lindblad channel (solid, filled markers) and the exact uniform local depolarising channel (dashed, open markers), divided by that of the global model (dotted), for ASPIRE (orange) and the MPD (blue). The two protocols are offset horizontally for legibility, and the error bars span the five rate draws.}
    \label{fig:noise-model-check}
\end{figure*}

\section{Details of the ASPIRE protocol}
\label{app:disentangling}

\subsection{Formal statement}
\label{app:formal}
Algorithm~\ref{alg:aspire} takes the target MPS $\ket{\psi}$ on $N$ qubits, the connectivity matrix $K$, the cost $c_{ij}=3(1+s_{ij})$ CNOTs of each admissible pair including its $s_{ij}$ routing SWAPs (Appendix~\ref{app:connectivity}), the absolute and relative thresholds $\tilde\tau$ and $\rho$, the native-gate budget $B$, the maximum layer size $m\le\lfloor N/2\rfloor$, the number $n_{\mathrm{c}}$ of highest-weight admitted pairs passed to the matching, and the truncation bond dimension $\chi_w$ to which the state is compressed after each gate. Each admissible pair is weighted by its removable QMI, either raw, $w_{ij}=R_{ij}$, or per native gate, $w_{ij}=R_{ij}/c_{ij}$, and enters the admitted set $E$ only if $w_{ij}\ge\tilde\tau$, if $w_{ij}$ is at least the fraction $\rho\in[0,1]$ of the largest current weight, and if $c_{ij}$ fits in the budget left after the $b$ native gates already spent. The loop acts on the working state $\ket{\phi}$ and records the applied gates in order in the list $\mathcal{G}$, whose inverses in reverse order form the circuit $C$. For $m=\lfloor N/2\rfloor$ and $n_{\mathrm{c}}\ge|E|$, line 13 is an unconstrained maximum-weight matching. Table~\ref{tab:selection} lists the settings behind each figure.
\begin{figure}\color{black}\algrenewcommand\alglinenumber[1]{{\color{black}\footnotesize#1:}}
\begin{algorithm}[H]
\caption{\color{black}ASPIRE (formal version of Algorithm~\ref{alg:state-prep})}\label{alg:aspire}
\begin{algorithmic}[1]
\Procedure{ASPIRE}{$\ket{\psi},K,\{c_{ij}\},\tilde\tau,\rho,B,m,n_{\mathrm{c}},\chi_w$}
    \State $\ket{\phi}\gets\ket{\psi}$, \ $b\gets0$, \ $\mathcal{G}\gets()$
    \Loop
        \ForAll{pairs $(i,j)$ with $K_{ij}=1$}
            \State $\rho_{ij}\gets$ two-site RDM of $\ket{\phi}$ \Comment{Equation~\ref{eqn:rdm_diagram_form}}
            \State $(U^*_{ij},R_{ij})\gets$ \Call{EigenbasisGate}{$\rho_{ij}$} \Comment{Proposition~\ref{prop:pair}}
            \State $w_{ij}\gets R_{ij}/c_{ij}$ (normalised) or $R_{ij}$ (raw) \Comment{$c_{ij}$ on the current layout}
        \EndFor
        \State $E\gets\{(i,j):K_{ij}=1,\ w_{ij}\ge\max(\tilde\tau,\rho\max_{kl}w_{kl}),\ b+c_{ij}\le B\}$
        \If{$E=\emptyset$}
            \State \textbf{break}
        \EndIf
        \State $\mathcal{L}\gets$ \Call{MaxWeightMatching}{$E,w,m,n_{\mathrm{c}}$} \Comment{$\le m$ disjoint pairs among the $n_{\mathrm{c}}$ best of $E$}
        \State drop the lowest-weight pairs of $\mathcal{L}$ until $b+\sum_{(i,j)\in\mathcal{L}}c_{ij}\le B$
        \ForAll{$(i,j)\in\mathcal{L}$}
            \State $\ket{\phi}\gets$ \Call{Compress}{$U^*_{ij}\ket{\phi},\chi_w$} \Comment{SWAP-routed if $|i-j|>1$}
            \State append $U^*_{ij}$ to $\mathcal{G}$, \ $b\gets b+c_{ij}$
        \EndFor
        \State record $T=\sum_aS(\rho^a)$ of $\ket{\phi}$
    \EndLoop
    \State \Return $C=V_1^\dagger V_2^\dagger\cdots V_M^\dagger$ for $\mathcal{G}=(V_1,\dots,V_M)$ \Comment{single-qubit layer of Equation~\ref{eq:infidelity-frame} applied first}
\EndProcedure
\end{algorithmic}
\end{algorithm}
\end{figure}
\begin{table*}[t]
\caption{Selection settings per figure. Step (III) of Algorithm~\ref{alg:state-prep} builds each layer from a matching, a set of pairwise disjoint pairs. Exact matching finds the maximum-weight matching by branch and bound over the $n_{\mathrm{c}}=32$ highest-weight admitted pairs, and blossom finds it over all admitted pairs with Edmonds' blossom algorithm (Section~\ref{sec:placement}). Greedy selection takes pairs in descending order of score and keeps each one that is disjoint from those already chosen. A cross pair $(a,b)$ joins a qubit of the candidate to a qubit already chosen. The score is the weight $w_{ij}$, the layer size the cap $m$ (full meaning $m=\lfloor N/2\rfloor$) and the admission rule the thresholds $\tilde\tau$ and $\rho$ (none meaning $\tilde\tau=\rho=0$) of Algorithm~\ref{alg:aspire}. In Figure~\ref{fig:fig2-randomMPS} the factor $(2|i-j|+1)^{\gamma}$ penalises long gates, which cost more to optimise, and $\gamma=0$ recovers $R_{ij}$.}
\label{tab:selection}
\color{black}
\renewcommand{\arraystretch}{1.15}
\setlength{\tabcolsep}{4pt}
\footnotesize
\begin{ruledtabular}
\begin{tabular}{lp{5.3cm}p{3.6cm}p{2.4cm}p{3.7cm}}
Figure(s) & Step-(III) rule & Score & Layer width & Admission rule \\
\hline
\ref{fig:overview} & blossom on the exact state vector, all $66$ pairs (schematic) & $R_{ij}$ & full; weakest gate of layer $2$ dropped & none \\
\ref{fig:toy-expressivity} & blossom for the GHZ $r=\infty$ curve, SWAP-routed beyond distance $8$, and both Bell-lattice arms; other curves hand-built & $R_{ij}/c_{ij}$, $c_{ij}=1+s_{ij}$ native gates & full & $\tilde\tau=10^{-3}$ \\
\ref{fig:fig1-efficacy} & exact matching at $r=2L_x-1$ & $R_{ij}$ & $\le\min(16,N/2)$ & none \\
\ref{fig:fig2-randomMPS} & greedy disjoint selection & $R_{ij}/(2|i-j|+1)^{\gamma}$, $\gamma\in\{0,\tfrac14,\tfrac12\}$ & $\le4$ or $1$ & $R_{ij}>10^{-12}$; $\gamma$ and width best per budget \\
\ref{fig:fig3-connectivity-spectra}, \ref{fig:fig3-connectivity-trends} & greedy, cross pairs $I(a{:}b)\le0.1$ bits & $I(i{:}j)$ & disjointness only & $\tilde\tau=10^{-4}$, $\rho=0.75$, lifted if the layer would be empty \\
\ref{fig:fig4-heavyhex} & exact matching over the $32$ couplers & $R_{ij}$ & $\le15$ & none; in \textbf{(b)} also $\rho=0.9$ or $\tilde\tau=10^{-3}$ on $R_{ij}/c_{ij}=R_{ij}/3$ \\
\ref{fig:fig5-quantics}, \ref{fig:figA-quantics} & exact matching & $R_{ij}$ & $\le\min(16,N/2)$ & none \\
\ref{fig:decomposition} & blossom & $R_{ij}$, or $R_{ij}/c_{ij}$ in the cost-aware arm & $1$ or full, per line & $\rho\in\{0,0.25,0.5,0.75,0.9\}$, per line \\
\ref{fig:toy-energy}, \ref{fig:qpe-triangle} & no selector & & & \\
\end{tabular}
\end{ruledtabular}
\end{table*}

\subsection{Exact evaluation of the score and the gate}
\label{app:gate}

The joint entropy $S(\rho^{ij})$ is invariant under unitaries on $\{i,j\}$, so it cancels between the two mutual informations and
\begin{equation}
  R_{ij}=S(\rho^{i})+S(\rho^{j})-\min_{U_{ij}}
  \bigl[S((\rho')^{i})+S((\rho')^{j})\bigr].
  \label{eq:Rmarginal}
\end{equation}
For two parties this minimisation was solved by Jevtic, Jennings and Rudolph~\cite{jevtic2012maximally,jevtic2012quantum}: the minimum is attained at a state diagonal in a product basis, and for two qubits it takes the closed form that Proposition~\ref{prop:pair} and Corollary~\ref{cor:sat} express in the present notation. Their argument rests on the reduction of the convex hull of a unitary orbit to classical distributions majorised by its spectrum~\cite{bravyi2003requirements}, whose two-qubit marginal conditions, among them $\mu_i+\mu_j\ge p_2+p_3+2p_4$, yield the same closed form. We claim novelty only for the extension of this argument in Theorem~\ref{thm:assign} to supports of any size and to any concave, permutation-invariant $F$, and for the use of the resulting quantity as a selection rule. Since ASPIRE scores every admissible pair at every layer, this minimisation must be exact and fast, and Theorem~\ref{thm:assign} makes it so by replacing the search over unitaries with a finite search over assignments of the spectrum to a product basis.

\begin{theorem}[Assignment reduction]
\label{thm:assign}
Let $A$ be a support with local dimensions $(d_a)_{a\in A}$, total dimension
$D=\prod_{a\in A}d_a$ and product index set $X=\prod_{a\in A}\{1,\dots,d_a\}$.
Let $\lambda$ be the spectrum of $\rho^{A}$, let $\lambda\circ\sigma$ denote the
probability vector on $X$ obtained from a bijection
$\sigma:\{1,\dots,D\}\to X$, and let $m_a$ be the site-$a$ marginal. Then for
every concave, permutation-invariant $F$,
\begin{equation}
  \min_{U\in\mathrm{U}(D)}\sum_{a\in A}F\bigl((U\rho^{A}U^{\dagger})^{a}\bigr)
  \;=\;\min_{\sigma}\sum_{a\in A}F\bigl(m_a(\lambda\circ\sigma)\bigr),
  \label{eq:assign}
\end{equation}
and a minimising unitary may always be chosen to map an eigenbasis of $\rho^{A}$
onto a product basis.
\end{theorem}

\noindent The proof is in Appendix~\ref{app:assignproof}. The score and the optimal gate therefore follow from one eigendecomposition and a discrete relabelling, and for a pair of qubits the assignment has a closed form.

\begin{proposition}[Two-qubit closed form]
\label{prop:pair}
Let $p_1\ge p_2\ge p_3\ge p_4$ be the eigenvalues of $\rho^{ij}$ and
$\mu_a=1-\lambda_{\max}(\rho^{a})$. Then
$R_{ij}=h_2(\mu_i)+h_2(\mu_j)-h_2(p_2+p_4)-h_2(p_3+p_4)$, attained by mapping the
eigenvectors of $\rho^{ij}$, ordered by decreasing eigenvalue, onto
$\ket{00},\ket{01},\ket{10},\ket{11}$.
\end{proposition}

\begin{proof}
Apply Theorem~\ref{thm:assign} with $F=S$ and $D=4$. The value of an assignment depends
only on the two induced marginals, and the row split $\{00,01\}|\{10,11\}$ and
column split $\{00,10\}|\{01,11\}$ each partition $\{p_1,p_2,p_3,p_4\}$ into two
pairs, the two partitions always being distinct. A four-element set admits
exactly three pairings, with smaller-part sums $p_3+p_4$, $p_2+p_4$ and
$\min\{p_1+p_4,p_2+p_3\}$, each at most $1/2$. Every assignment realises an
ordered pair of two distinct entries of this list, and every such pair is
realised. Since the three entries are ordered as listed and $h_2$ is strictly
increasing on $[0,\tfrac12]$, the sum is minimised by the two smallest, which is
what the sorted assignment realises.
\end{proof}

\begin{corollary}[Saturation of the pair score]
\label{cor:sat}
Arrange the sorted spectrum as the table
$P=\bigl(\begin{smallmatrix}p_1&p_2\\p_3&p_4\end{smallmatrix}\bigr)$. Then
$I(i{:}j)-R_{ij}=I_{\mathrm{cl}}(P)$, the classical mutual information of that
table, and consequently $R_{ij}=I(i{:}j)$ if and only if $p_1p_4=p_2p_3$. Saturation therefore holds for every $\rho^{ij}$ of rank at most two, fails for every $\rho^{ij}$ of rank three, and holds only on a set of measure zero at full rank, which is the generic case for a pair inside a larger entangled state.
\end{corollary}

\begin{proof}
By Proposition~\ref{prop:pair} and the definition of $I(i{:}j)$, $I(i{:}j)-R_{ij}=h_2(p_3{+}p_4)+h_2(p_2{+}p_4)-S(\rho^{ij})$, the entropies of the row and column marginals of $P$ minus that of $P$ itself, which is $I_{\mathrm{cl}}(P)$. A two-by-two table has zero classical mutual information exactly when it has rank one, that is when $p_1p_4=p_2p_3$, which holds automatically when $p_3=p_4=0$ and never when $p_4=0<p_3$.
\end{proof}

A pair can be frozen at $R_{ij}=0$ while carrying $I(i{:}j)>0$: the state $\mathrm{diag}(0.5,0.3,0.15,0.05)$ has $I=8.3\times10^{-3}$ bits and $R_{ij}=0$. At the $66$-gate halt of the $r=1$ runs of Figure~\ref{fig:fig3-connectivity-spectra}, counting pairs with $R_{ij}<10^{-9}$, $I(i{:}j)\ge10^{-9}$ and $\rho^{ij}$ not maximally mixed finds $36$ of the $66$ pairs of the $4\times3$ Heisenberg ladder frozen on every seed, including all $11$ nearest-neighbour pairs, against a median of $3.5$ on the $3\times4$ TFIM. Tracking every pair with $I(i{:}j)>10^{-3}$ bits at $R_{ij}<10^{-6}$ through greedy runs capped at $60$ layers, on a $12$-qubit heavy-hex cell, a $2\times5$ Heisenberg snake and a random MPS with $N=10$ and $\chi=4$, $98\%$ of the roughly $300$ stall events per target are followed by the same pair becoming actionable again, after a median of one layer and at most four, and whether this recovery holds at $N>12$ is untested.

The optimal gate is fixed only up to the choice of basis inside a degenerate eigenvalue block of $\rho^{A}$, which symmetric targets such as the rung singlets of a Heisenberg ladder realise exactly. We identify degenerate blocks at relative tolerance $10^{-9}$ and, for a block of size $k>1$, apply the block eigenprojector to the computational basis in fixed order, orthonormalise the images by Gram--Schmidt and keep the first $k$ independent ones, so the gate depends on the eigenspace and not on the returned eigenvectors.

\subsection{Routing cost and the source of the advantage}
\label{app:connectivity}

The SWAP count $s_{ij}$ in the cost of Appendix~\ref{app:formal} is taken on the current layout. SWAPs that bring qubits $i$ and $j$ together are not undone afterwards, so the positions of the logical qubits on the device change as the circuit grows, and a pair that is $\ell$ couplers apart at that point needs $s_{ij}=\ell-1$ SWAPs and costs $c_{ij}=3\ell$ CNOTs. Every $c_{ij}$ is therefore recomputed after each layer, and a pair that has just been routed together costs only three CNOTs to act on again. No SWAP is merged into a gate on the same two qubits, although the product would compile to three CNOTs, so the count is conservative. On a coupler every pair costs the same three CNOTs, and ranking by $R_{ij}$ or by $\tilde R_{ij}$ gives the same order. The two rankings differ only for routed pairs, where $\tilde R_{ij}$ is the natural choice, because one gate lowers $T$ by at most two bits whatever the distance between its qubits while its cost grows linearly with that distance. \label{app:tau}It is not uniformly better, however. At an equal CNOT budget, ranking by $R_{ij}$ reached the lower $\varepsilon_\star$ in seven of nine $(N,\chi)$ configurations of the heavy-hex patch of Section~\ref{subsec:noise}, and Figure~\ref{fig:decomposition} shows the same split.

\label{app:routing}Figure~\ref{fig:decomposition} tracks the infidelity of the prepared state as a function of two-qubit gates applied,  separating the parts of ASPIRE's advantage over the MPD on the least favourable device, a linear chain whose couplers are the bonds of the MPS, so that every gate between sites that are not neighbours in the chain must be routed. All curves are charged in CNOTs and stop at the $198$ CNOTs of the MPD's $66$ nearest-neighbour gates, and each adds one ingredient to the one before. The green curve applies ASPIRE's gate at exactly the positions and in the order of the MPD. It is worse than the MPD on all three targets, by factors of about six and ten on the two lattices, so the gate on its own is not the source of the advantage. The purple curve lets ASPIRE choose its own pairs among nearest neighbours and closes most of that gap, but still ends above the MPD. The solid orange curve lets it choose any pair and pays for the SWAPs, and the dashed curve does the same but ranks pairs by $\tilde R_{ij}$, with each $c_{ij}$ taken from the initial layout rather than the current one. The dotted curve runs the all-pairs selection but charges no SWAPs, as if every pair were a coupler, which no chain can realise, so the gap between it and the orange curve measures what routing costs. Measured as the MPD's infidelity divided by ASPIRE's at the MPD's budget, the solid orange curve gains $1.22$ and $1.17$ on the Heisenberg and TFIM targets, the dashed curve $1.14$ and $1.99$, and the dotted curve $2.7$ and $8.8$. On the random MPS every curve except the dotted one ends above the MPD, the orange one at $0.75$, and even free connectivity gains only $1.15$. Long-range selection therefore pays for its routing on the two lattices, whose geometry the snake ordering hides from nearest-neighbour gates, and has little to find on a state without such structure. Each ASPIRE curve is drawn at the relative threshold $\rho\in\{0,0.25,0.5,0.75,0.9\}$ and the layer size, one gate or a full matching, that is best for its own panel, so the figure shows the best of these settings rather than a single run. With $\rho=0.75$ and full layers throughout, the solid orange curve gains $1.07$ and $1.09$ on the lattices, a tie with the MPD, and without a threshold it falls to $0.66$ and $0.33$, so the threshold is what lets long-range gates pay for their SWAPs.

\begin{figure*}
    \centering
    \includegraphics[width=0.48\linewidth]{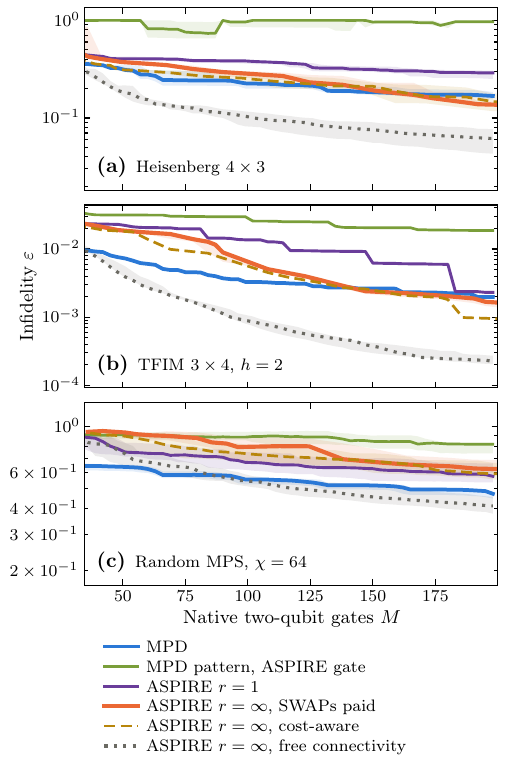}
    \caption{\textbf{Where the advantage comes from.}
    Infidelity against native two-qubit gates for \textbf{(a)} the $4\times3$
    Heisenberg ladder, \textbf{(b)} the $3\times4$ TFIM at $h=2$, and \textbf{(c)} a
    random MPS at $\chi=64$, all at $N=12$. Each line is the median over $16$ seeds,
    and its shaded band covers the middle half of the seeds. The MPD (blue) is
    compared with ASPIRE in five settings. ASPIRE's gate placed at the MPD's positions and in its order (green)
    changes only the gate. ASPIRE restricted to nearest-neighbour pairs, $r=1$
    (purple), adds the selection rule at the MPD's reach. ASPIRE over all pairs on a
    nearest-neighbour device, paying for the SWAPs that bring each pair together
    (solid orange), adds the reach, and the same with pairs ranked by the correlation
    removed per native gate (dashed gold) makes that ranking cost-aware. ASPIRE over
    all pairs on a device that couples every pair (dotted grey) needs no SWAPs and
    shows what the reach is worth without routing. Each ASPIRE line uses the relative
    threshold and layer size that are best for its panel (Appendix~\ref{app:routing})}
    \label{fig:decomposition}
\end{figure*}

\subsection{Layer generation}
\label{app:layer}

Step (III) is an exact maximum-weight matching, computed by Edmonds' algorithm~\cite{edmonds1965paths} over all admitted pairs or by branch and bound over the $32$ highest-weight ones (Table~\ref{tab:selection}). In a layer size sweep at the MPD's six-layer budget, with one seed per (target, $\chi$) cell, a narrower layer beat the full matching in all $3$ Heisenberg cells and both random-MPS cells, one of the latter by only $4\%$, and the full matching won both TFIM cells. At a four-layer budget over four seeds, layer size two lost to the full matching in $9$ of $13$ cells, including $3$ of the $4$ Heisenberg cells, so the preferred layer size depends on the budget. The layer-size cap $m$ also has a hardware reading: shuttling-based trapped-ion processors realise all-to-all connectivity by moving ions rather than by fixed couplers, but apply two-qubit gates only in a few dedicated zones, with the Quantinuum H2 executing up to four per round \cite{moses2023race} and Helios using four of its eight operation zones for them~\cite{ransford202698}. This constraint and others like it can then be straightforwardly accounted for by ASPIRE.

\subsection{Gate support size}
\label{app:arity}

The restriction to two-site gates is a cost argument. The removable correlation of a support obeys $R_A\le|A|\log d$, which is linear in $|A|$, whereas a generic $k$-qubit unitary costs $c_k=\tfrac{23}{48}4^{k}-\tfrac32 2^{k}+\tfrac43$ CNOTs by quantum Shannon decomposition~\cite{shende2005synthesis}, so the ceiling per CNOT at QSD cost is
\begin{equation}
  \eta_k=\frac{k\log d}{c_k}=\mathcal{O}(k\,4^{-k}),
  \label{eq:arity}
\end{equation}
namely $0.667$, $0.150$, $0.040$ and $0.011$ bits per CNOT at $k=2,\dots,5$.

\subsection{Numerical details of the targets}
\label{app:targets}

All states and circuits are simulated as MPS with ITensors.jl and ITensorMPS.jl in Julia \cite{fishman2022itensor}, and the ground-state targets are found by DMRG with $8$ to $16$ sweeps, a bond dimension ramped up to the cap each figure states and a truncation cutoff of $10^{-12}$. Infidelities are measured against the resulting MPS itself (Equation~\ref{eq:infidelity}), so the DMRG truncation error enters no reported number, and circuit simulation truncates at cutoff $10^{-10}$ to $10^{-12}$ with a bond-dimension ceiling between $64$ and $4096$ depending on the experiment. That separate evaluation runs at the same bond-dimension ceiling as the construction. The ceiling is never below $64=2^{6}$, so it cannot bind for $N\le13$, which covers every $N=12$ study, and those infidelities are exact up to the cutoff. For the $10\times10$ lattice of Figure~\ref{fig:fig1-efficacy} the ceiling is $\chi=512$ against a target of bond dimension at most $200$, and it binds for both circuits. The infidelities are nonetheless insensitive to it when the discarded weight is counted against the fidelity, as it is for ASPIRE: re-evaluating at $\chi=1024$ changes the MPD's final infidelity by less than $10^{-5}$ in relative terms, and ASPIRE's at $156$ gates by $4\times10^{-6}$, although the discarded weight itself falls by a factor of three to five. The MPD's infidelity in Figure~\ref{fig:fig1-efficacy} is instead evaluated with the residual renormalised after each truncation, which lowers it by $0.8\%$ relative to that convention, so the MPD's quoted value is if anything slightly optimistic.

\section{Proofs}
\label{app:proofs}

\subsection{Exact descent and additivity}
\label{app:propproof}

\begin{proof}[Proof of Proposition~\ref{prop:descent}]
Let $U_{\mathcal L}=\bigotimes_gU_{A_g}$ with the $A_g$ pairwise
disjoint. For $a\notin\bigcup_g A_g$ every site of every $A_g$ is traced out in
$\rho^{a}$, so $\rho^{a}$ is invariant and contributes nothing to $\Delta T$. For
$a\in A_g$ the marginal $(\rho')^{a}$ is determined by $U_{A_g}$ and
$\rho^{A_g}$ alone, because the remaining factors of $U_{\mathcal L}$ act on
sites disjoint from $A_g$ and are traced out. Hence
$\Delta T=\sum_g\Delta_g$ with
$\Delta_g=\sum_{a\in A_g}[S((\rho')^{a})-S(\rho^{a})]$, exactly and with no cross
term at any order. The inequality $\Delta_g\ge-R_{A_g}$ is the definition of
$R_{A_g}$ as the maximal achievable reduction on $A_g$, evaluated on the
pre-layer state, with equality if and only if $U_{A_g}$ attains the
minimum. By Theorem~\ref{thm:assign} that minimum is attained by an explicit gate, so
the equality case is realisable and the bound is not conditional on a numerical
optimisation.
\end{proof}

\subsection{The infidelity bound}
\label{app:thm2}

\begin{proof}[Proof of Theorem~\ref{thm:cert}]
For each site choose an orthonormal eigenbasis $\{\ket{k}_a\}_{k=0}^{d-1}$ of $\rho^{a}_{\phi}$ with $\ket{0}_a$ a top eigenvector, and expand $\ket\phi=\sum_x c_x\ket{x}$ in the product basis $\ket{x}=\bigotimes_a\ket{x_a}_a$. Since $\bigotimes_j v_j\ket{0}^{\otimes N}$ ranges over all product states, $f_\star=\max_{\ket v\ \text{product}}|\braket{v}{\phi}|\ge|c_0|$, so $1-f_\star^{2}\le\delta$ with $\delta=1-|c_0|^{2}$. Let $|x|$ be the number of sites with $x_a\neq0$. Since $1-\lambda_a=\sum_{x:\,x_a\neq0}|c_x|^{2}$,
\begin{equation*}
  \Lambda=\sum_x|x|\,|c_x|^{2},\qquad \delta=\sum_{x\neq0}|c_x|^{2}.
\end{equation*}
Write $k e_a$ for the string equal to $k\neq0$ at site $a$ and to $0$ elsewhere, and let $X_a=\sum_{k\neq0}|c_{k e_a}|^{2}$ and $W_1=\sum_a X_a$ be the weight on strings with $|x|=1$. Every other nonzero string has $|x|\ge2$, so $\Lambda\ge W_1+2(\delta-W_1)=2\delta-W_1$. The weight $W_1$ is controlled by the diagonality of the marginals. Write $(j,r)$ for the string with $j$ at site $a$ and configuration $r$ on the other sites. For $k\neq0$, $0=\bra{0}\rho^{a}_{\phi}\ket{k}=\sum_r c_{(0,r)}c^{*}_{(k,r)}$, whose $r=0$ term is $c_0c^{*}_{k e_a}$, and the Cauchy--Schwarz inequality applied to the remaining terms gives
\begin{equation*}
  |c_0|^{2}\,|c_{k e_a}|^{2}\ \le\ \sum_{r\neq0}|c_{(0,r)}|^{2}\ \sum_{r\neq0}|c_{(k,r)}|^{2}\ \le\ \delta\sum_{r\neq0}|c_{(k,r)}|^{2} .
\end{equation*}
Summed over $k\neq0$ this reads $(1-\delta)X_a\le\delta(1-\lambda_a-X_a)$, that is $X_a\le\delta(1-\lambda_a)$, so $W_1\le\delta\Lambda$. Hence $\Lambda\ge\delta(2-\Lambda)$, so for $\Lambda<2$, $1-f_\star^{2}\le\delta\le\tilde\Lambda$, which is the first inequality of Equation~\ref{eq:cert} whenever $\tilde\Lambda\le1$, that is whenever $\Lambda\le1$. The second is $\sqrt{1-\tilde\Lambda}\ge1-\tilde\Lambda$ and the third is $2-\Lambda\ge1$. For the last, if $\lambda_a\ge\tfrac12$, binning all non-maximal eigenvalues of $\rho^{a}_{\phi}$ together can only decrease the Shannon entropy of the spectrum, so $S(\rho^{a}_\phi)\ge h_2(1-\lambda_a)\ge2(1-\lambda_a)$, the second step because $h_2$ is concave with $h_2(0)=0$ and $h_2(\tfrac12)=1$. If $\lambda_a<\tfrac12$, which requires $d>2$, then $S(\rho^{a}_\phi)\ge-\log_2\lambda_a\ge2(1-\lambda_a)$, the second step because $-\log_2\lambda-2(1-\lambda)$ vanishes at $\lambda=\tfrac12$ and decreases on $(0,\tfrac12]$. Summing over $a$ gives $2\Lambda\le T$.
\end{proof}

The first inequality of Equation~\ref{eq:cert} is tight to leading order. The state $\sqrt{1-\delta}\ket{00}+\sqrt{\delta}\ket{11}$ with $\delta<\tfrac12$ has $\Lambda=2\delta$ and $f_\star^{2}=1-\delta$, so $\varepsilon_\star=\Lambda/4+O(\Lambda^{2})$ and the coefficient cannot be lowered, and the $N$-qubit $W$ state with $N\ge3$ attains $\delta=\tilde\Lambda=1$. The gain over the union bound $\delta\le\Lambda$ comes from the diagonality step: strings that leave the top eigenvector at a single site carry weight at most $\delta\Lambda$, so every other residual component is counted at least twice in $\Lambda$. The last inequality is loose by a factor that grows as the circuit improves, since $h_2(\epsilon)\ge\epsilon\log_2(1/\epsilon)$ gives $T\ge\Lambda\log_2\bigl(1/\max_a(1-\lambda_a)\bigr)$, and $T$ alone bounds $\varepsilon_\star$ only once it has fallen below two bits.

\subsection{The assignment reduction}
\label{app:assignproof}

\begin{proof}[Proof of Theorem~\ref{thm:assign}]
For the upper bound, fix a product basis $\{\ket{e_x}\}_{x\in X}$ and an
eigenbasis $\{\ket{v_k}\}$ of $\rho^{A}$. For a bijection $\sigma$ let
$U_\sigma\ket{v_k}=\ket{e_{\sigma(k)}}$. Then
$U_\sigma\rho^{A}U_\sigma^{\dagger}=\sum_k\lambda_k\dyad{e_{\sigma(k)}}$
is diagonal in the product basis with distribution $\lambda\circ\sigma$, the
vector placing $\lambda_k$ at $\sigma(k)$, and its
site-$a$ marginal has eigenvalue vector $m_a(\lambda\circ\sigma)$. The left side
of Equation~\ref{eq:assign} is therefore at most the right.

For the lower bound, let $U$ be arbitrary and $\rho'=U\rho^{A}
U^{\dagger}$. Choose an eigenbasis $\{\ket{e^a_j}\}$ of each marginal
$(\rho')^{a}$, form $\ket{e_x}=\bigotimes_a\ket{e^a_{x_a}}$, and set
$p(x)=\bra{e_x}\rho'\ket{e_x}$. By the Schur--Horn theorem the diagonal of a
Hermitian operator in any orthonormal basis is majorised by its spectrum, so
$p\prec\lambda$; and $m_a(p)$ is the eigenvalue vector of $(\rho')^{a}$ by
construction, so $F((\rho')^{a})=F(m_a(p))$. By the Hardy--Littlewood--P\'olya
theorem together with the Birkhoff--von~Neumann decomposition, $p\prec\lambda$
implies $p=\sum_\pi t_\pi(\lambda\circ\pi)$ with $t_\pi\ge0$ summing to one.
Since $m_a$ is linear and $F$ is concave,
$F(m_a(p))\ge\sum_\pi t_\pi F(m_a(\lambda\circ\pi))$; summing over $a$ and
exchanging the two finite sums gives
$\sum_{a}F((\rho')^{a})\ge\min_\pi\sum_{a}F(m_a(\lambda\circ\pi))$. As $U$
was arbitrary, the two bounds coincide, and $U_\sigma$ at a minimising $\sigma$
is a minimising unitary.
\end{proof}

\subsection{Correlation removal at bounded range}
\label{app:ghzbound}

\begin{proposition}[Bounded-range correlation removal]
\label{prop:ghzbound}
Let $U$ be a circuit of depth $D$ whose layers consist of pairwise disjoint two-qubit gates of range at most $r$, together with arbitrary single-qubit gates, and let $\ket\phi=U\ket{\mathrm{GHZ}_N}$.

(i) If $N>2rD+1$, then $T(\phi)\ge\lceil N/(2rD+1)\rceil$.

(ii) If $\ket\phi$ is a product state, then $2^{D}\ge N$ and $2rD\ge N-1$, so exact preparation of $\mathrm{GHZ}_N$ needs $D\ge\max\{\lceil\log_2N\rceil,\lceil(N-1)/(2r)\rceil\}$. The first bound is attained at $r=\infty$ by the CNOT tree of Appendix~\ref{app:ghz}.
\end{proposition}

\begin{proof}
A layer at most doubles the support of an operator and widens it by $r$ on each side, so the cone $C_a=\operatorname{supp}(\tilde Z_a)$ of $\tilde Z_a=UZ_aU^{\dagger}$ has $|C_a|\le2^{D}$ and $C_a\subseteq[a-rD,a+rD]$.

(i) The $\lceil N/(2rD+1)\rceil\ge2$ sites spaced $2rD+1$ apart have pairwise disjoint cones. For two of them, $a$ and $b$, $\tilde Z_a\tilde Z_b$ stabilises $\ket\phi$ and $\tilde Z_b$ acts outside $C_a$, so $\rho^{C_a}$ has no coherence between the $\pm1$ eigenspaces of $\tilde Z_a$. Both have weight $\tfrac12$ because $\langle\tilde Z_a\rangle_\phi=0$, so $S(\rho^{C_a})\ge1$ bit, and subadditivity over the disjoint cones gives the bound. This is the standard light-cone argument of \cite{bravyi2006lieb}.

(ii) A product state has $T=0$, so $2rD\ge N-1$ follows from (i). For the other bound, write the product state as $\bigotimes_jv_j\ket{0^{\otimes N}}$. Then $V=U^{\dagger}(\bigotimes_jv_j)$ prepares $\ket{\mathrm{GHZ}_N}$ from $\ket{0^{\otimes N}}$ with the same depth, and the operators $O_a=VZ_aV^{\dagger}$, whose supports $S_a$ have $|S_a|\le2^{D}$, have $\ket{\mathrm{GHZ}_N}$ as their only joint $+1$ eigenvector. If every $S_a$ were a proper subset, projecting $O_a\ket{\mathrm{GHZ}_N}=\ket{\mathrm{GHZ}_N}$ onto $\bra{0_{\bar S_a}}$ and $\bra{1_{\bar S_a}}$ would show that $O_a$ fixes $\ket{0_{S_a}}$ and $\ket{1_{S_a}}$, and hence also $(\ket{0^{\otimes N}}-\ket{1^{\otimes N}})/\sqrt2$. So some $S_a$ is the whole chain, and $2^{D}\ge N$.
\end{proof}

\noindent At $N=64$ and $r=1$, (i) holds for every $D\le31$, which covers the shaded region of Figure~\ref{fig:toy-expressivity}\textbf{(a)}.

\section{Gate counts on quantics function states}
\label{app:quantics}

Figure~\ref{fig:figA-quantics} complements the matched-budget comparison of Section~\ref{subsec:quantics} by reporting the gates each protocol needs to reach a stated infidelity, on a correlated bivariate Gaussian at $N=2B=10$ to $20$, run to $48$ two-qubit gates per site. At $\varepsilon=10^{-3}$ the saving grows from $1.8\times$ to $3.6\times$ in the serial ordering and from $2.5\times$ to $4.7\times$ in the interleaved one, with fitted slopes of $+0.025$ to $+0.030$ decades per site, and at $\varepsilon=10^{-2}$ the interleaved saving grows from $3.9\times$ to $7.1\times$. The advantage therefore grows with the number of precision bits, as it does at the matched budget of Section~\ref{subsec:quantics}. The separation is also larger in the interleaved ordering, because interleaving places each variable's own significance chain at range $D$, off the pairs the MPD can reach, which is consistent with the ordering sensitivity of Section~\ref{subsec:quantics}. Both arms are unoptimised, as in Figure~\ref{fig:fig5-quantics}.

\begin{figure*}
    \centering
    \includegraphics[width=\linewidth]{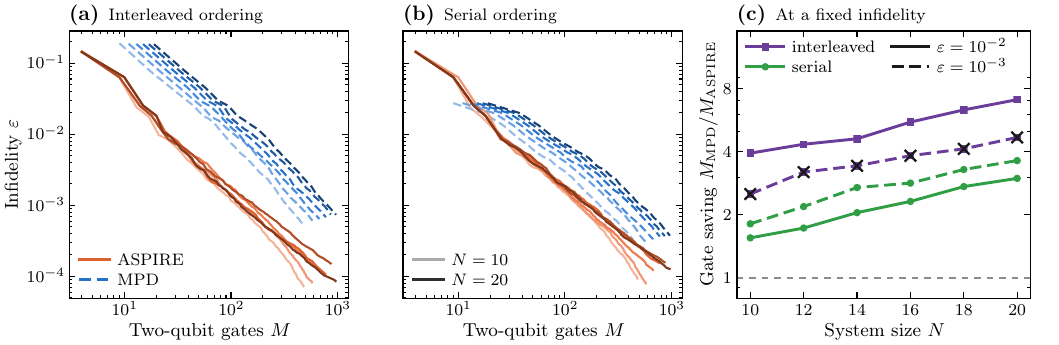}
    \caption{\textbf{Quantics function states: the gate saving at a stated infidelity.} A correlated bivariate Gaussian on $N=2B$ qubits, $8$ seeds, both arms unoptimised and run to $48$ two-qubit gates per site. \textbf{(a), (b)} Infidelity against two-qubit gate count, one per gate for both arms at all-to-all connectivity, for the interleaved and serial bit orderings, ASPIRE solid (orange) and MPD dashed (blue), with darker shades for larger $N=10,\dots,20$. \textbf{(c)} The gates the MPD needs to reach a stated infidelity divided by the gates ASPIRE needs, against $N$, at two infidelity levels and in both orderings; crosses mark cells in which not every seed reached the target. The saving grows monotonically with $N$ in every series.}
    \label{fig:figA-quantics}
\end{figure*}

\section{Stabiliser targets}
\label{app:stab}

Stabiliser states can be analysed exactly, because the marginal of a stabiliser state on a support $A$ is $\rho^{A}=2^{-|A|}\sum_{g}g$, summed over the stabiliser elements $g$ supported in $A$. How much a pair scores is therefore fixed by which stabiliser elements it carries, and the two extremes of the pair score are both stabiliser states. On $\mathrm{GHZ}_N$ every pair carries $Z_iZ_j$ and scores the maximal one bit, which gives the exact logarithmic-depth preparation of Section~\ref{sec:protocol}. On a state whose non-identity stabiliser elements all have weight at least $w\ge3$, no pair carries any, so every pair scores zero and Algorithm~\ref{alg:aspire} halts without emitting a gate, the stall of Section~\ref{sec:limitations}. Appendix~\ref{app:ghz} treats the two extremes, and Appendix~\ref{app:composite} the composite moves that remove the stall.

\subsection{The two extremes}
\label{app:ghz}

For $N\ge3$ every two-site marginal of $\mathrm{GHZ}_N$ is $\rho_{ij}=(\dyad{00}+\dyad{11})/2$, which a CNOT maps to the product of a pure and a maximally mixed qubit, so $R_{ij}=1$ bit at every distance, the most that $S(\rho_{ij})=1$ allows by subadditivity. A CNOT tree prepares $\mathrm{GHZ}_N$ in $\lceil\log_2N\rceil$ layers and $N-1$ gates, which is depth-optimal by Proposition~\ref{prop:ghzbound}(ii). At radius $r$ the tree can run until the active sites are $r$ apart, giving depth $\mathcal{O}(N/r+\log_2r)$ against the lower bound $\lceil(N-1)/(2r)\rceil$ of the same proposition. Greedy matching does not find this schedule at $r=1$: at $N=16$ it halts after one layer with $8$ of $16$ bits left, which is why the finite-radius curves of Figure~\ref{fig:toy-expressivity}\textbf{(a)} follow the schedule itself (Table~\ref{tab:selection}).

At the other extreme, a support $A$ with $|A|<w$ carries only the identity, so $\rho^{A}=\mathbf{1}/2^{|A|}$, every unitary on $A$ leaves it unchanged, and $R_A=0$. For $w\ge3$ this holds for every pair at every distance. The one-dimensional cluster state is the simplest example. On a ring of $N\ge5$ sites its stabilisers $K_i=Z_{i-1}X_iZ_{i+1}$ and all their products have weight at least three, so every pair scores zero. On an open chain the boundary stabilisers $X_1Z_2$ and $Z_{N-1}X_N$ have weight two, so pair-only selection can act only at the two ends, working inward in $\Theta(N)$ layers.

\subsection{Composite moves}
\label{app:composite}
Composite moves remove this stall with the native gate set: several two-qubit gates on overlapping pairs, scored and applied as one move. A generic gate on the whole support would also remove it, but at a cost growing as $4^{k}$ in the support size $k$ (Appendix~\ref{app:arity}).

\begin{definition}[Composite move]
\label{def:composite}
Let $A$ be a set of $|A|\ge3$ sites and $E=(e_1,\dots,e_\ell)$ an ordered list of admissible pairs $e_m\subset A$ such that the graph $(A,E)$ is connected. A composite move on $(A,E)$ is a unitary $V=U_{e_\ell}\cdots U_{e_1}$, each $U_{e_m}$ acting on $e_m$, of cost $c_{A,E}=\sum_m c_{e_m}$ and composite score
$$
R_A^E = T_A - \min_{U_{e_1},\dots,U_{e_\ell}} \sum_{a\in A} S\big((V\rho_AV^\dagger)_a\big),
\qquad
T_A=\sum_{a\in A}S(\rho_a).
$$
\end{definition}

\begin{proposition}[Composite score]
\label{prop:composite-bounds}
(i) Proposition~\ref{prop:descent} holds for any layer of pairwise disjoint pair and composite moves: the reduction of $T$ is the sum of the reductions of the moves, and a composite move on $(A,E)$ reduces $T$ by at most $R_A^E$.
(ii) $\max_m R_{e_m}\le R_A^E\le R_A\le T_A-S(\rho_A)$.
\end{proposition}

\begin{proof}
(i) Proposition~\ref{prop:descent} holds for arbitrary unitaries on pairwise disjoint supports, and a composite move is a unitary on $A$. (ii) Setting every factor but $U_{e_m}$ to the identity gives the first inequality, and the second holds because $V$ is a unitary on $A$. The third is subadditivity: $\sum_{a\in A}S\big((V\rho_AV^\dagger)_a\big)\ge S(V\rho_AV^\dagger)=S(\rho_A)$.
\end{proof}

\noindent The first inequality of (ii) can be strict with every pair at zero. On the bulk cluster triple $A=\{i-1,i,i+1\}$ the only non-identity stabiliser element is $K_i$, so every pair in $A$ scores zero, while $\mathrm{CZ}_{i-1,i}\mathrm{CZ}_{i,i+1}$ maps $K_i$ to $X_i$ and attains $R_A^E=1$ bit, the most that $R_A\le T_A-S(\rho_A)=1$ allows, with two CZ gates, costing $c_{A,E}=6$ CNOTs as generic blocks and two as compiled, against $20$ for a generic three-qubit unitary. The same repair works on every stabiliser state with $T>0$. Some stabiliser element is then supported only on maximally mixed sites, and a tree of CNOTs along its support, with single-qubit Clifford gates and SWAPs where its pairs are not couplers, maps it to a single-site $Z$ and purifies that site, so at most $N$ composite moves prepare the state exactly. Stabiliser states are in any case prepared efficiently by Clifford synthesis~\cite{aaronson2004improved}, so the point of the construction is that the stall is a limitation of pair moves and not of the selection rule.

On general targets composite moves are harder to use. The minimisation in Definition~\ref{def:composite} runs over a proper subset of $\mathrm{U}(2^{|A|})$ unless $E$ is long enough to generate it, so Theorem~\ref{thm:assign} does not apply and we know of no closed form. The connected supports of size $k$ through a site also grow exponentially in $k$ on any two-dimensional lattice. With $k-1$ generic blocks along a spanning tree and $R_A^E\le k$ bits, the ceiling per CNOT is $\eta_k^{\mathrm{comp}}=k/[3(k-1)]$, that is $0.500$, $0.444$ and $0.417$ bits at $k=3,4,5$. This is far above the $0.150$, $0.040$ and $0.011$ of the generic gates of Equation~\ref{eq:arity} but below the pair ceiling $\eta_2=2/3$, so on couplers a composite move on $k$ sites can beat a pair per CNOT only where the best pair scores below $k/(k-1)$ bits, as at a stall. Disjoint pair and composite moves still add by Proposition~\ref{prop:composite-bounds}(i), but choosing a maximum-weight disjoint family of supports of size up to $k\ge3$ is in general weighted $k$-set packing, which is NP-hard~\cite{johnson1979computers}. A greedy layer in descending score attains at least a fraction $1/k$ of the optimum, since each support of an optimal family meets a chosen support of no smaller score, the first one chosen that meets it, and each chosen support has at most $k$ sites and so meets at most $k$ of the disjoint optimal supports. If every candidate is an interval of the chain the selection is weighted interval scheduling, solved exactly in $\mathcal{O}(n\log n)$ time~\cite{kleinberg2006algorithm}.

Register-wide stalls are not confined to stabiliser states. For $x\in\mathbb{F}_2^{3}$ let $c(x)\in\mathbb{F}_2^{7}$ collect the parities $x\cdot u$ over the seven nonzero $u\in\mathbb{F}_2^{3}$, one site per $u$, so that the strings $c(x)$ form the $[7,3,4]$ simplex code, and let
\begin{equation}
  \ket{\psi}=\tfrac{1}{\sqrt2}\ket{0^{\otimes7}}+\tfrac{1}{\sqrt{14}}\textstyle\sum_{x\neq0}\ket{c(x)} .
  \label{eq:simplex-stall}
\end{equation}
Its amplitudes differ in modulus, so it is not a stabiliser state. Distinct codewords differ on four sites, so no two terms of Equation~\ref{eq:simplex-stall} agree outside a pair $\{i,j\}$, and every $\rho_{ij}$ is diagonal in the computational basis. For $u_i\neq u_j$ the map $x\mapsto(x\cdot u_i,x\cdot u_j)$ takes each value twice, so $\rho_{ij}=\mathrm{diag}(4/7,1/7,1/7,1/7)$ on all $21$ pairs. This is already the sorted arrangement of Proposition~\ref{prop:pair}, so $R_{ij}=0$, while $I(i{:}j)=2h_2(2/7)-S(\rho_{ij})\approx0.0617$ bits by Corollary~\ref{cor:sat}. The degenerate-block convention of Appendix~\ref{app:gate} then returns the identity, so admitting zero-score moves does not help either, and Algorithm~\ref{alg:aspire} halts at $T\approx6.04$ bits and $\varepsilon_\star\approx0.29$, where $\Lambda=2$ and Theorem~\ref{thm:cert} gives no bound. The stall is not generic: in a statevector check, none of $40$ Haar-random seven-qubit states and none of $60$ random perturbations of Equation~\ref{eq:simplex-stall} halted before exact preparation. Two escapes remove it. Each site whose $u$ has weight at least two holds the parity of the sites whose $u$ has weight one, and the nine CNOTs that undo these parities form four composite moves in the sense of Definition~\ref{def:composite}, after which two pair gates prepare the state exactly. Without using the code structure, a single MPD layer of six gates raises $f_\star$ from $0.707$ to $0.756$ and restores pair scores of up to $0.65$ bits, after which pair selection proceeds normally.

\end{document}